\documentclass[11pt]{article}
\usepackage[english]{babel}
\usepackage[numbers,sort&compress]{natbib}
\usepackage{graphicx}
\usepackage[dvipsnames,table,xcdraw]{xcolor}
\usepackage{framed}
\usepackage[normalem]{ulem}
\usepackage{amsmath}
\usepackage{amsthm,thmtools,mathtools,thm-restate}
\usepackage{amssymb}
\usepackage{amsfonts}
\usepackage{bbm}
\usepackage{enumerate}
\usepackage{algpseudocode}
\usepackage[linesnumbered,ruled,vlined,noend]{algorithm2e}
\SetKwRepeat{Do}{do}{while}
\usepackage{nomencl}
\usepackage{enumitem}
\usepackage{wrapfig} 
\usepackage{caption}
\usepackage{subcaption}
\usepackage[margin=1in]{geometry}
\usepackage[
  pagebackref,
  colorlinks=true,
  urlcolor=blue,
  linkcolor=RoyalBlue,
  citecolor=OliveGreen,
]{hyperref}
\usepackage[nameinlink]{cleveref}
\theoremstyle{definition}
\newtheorem{theorem}{Theorem}
\newtheorem{observation}{Observation}
\newtheorem{corollary}{Corollary}
\newtheorem{lemma}{Lemma}
\newtheorem{proposition}{Proposition}
\newtheorem{definition}{Definition}

\newtheorem*{remark}{Remark}

\newtheorem{problem}{Problem}

\DeclareMathOperator*{\argmax}{arg\,max} 
\DeclareMathOperator*{\argmin}{arg\,min}
\DeclareMathOperator{\E}{\mathbb{E}}
\renewcommand{\vec}[1]{\mathbf{#1}}
\newcommand{\points}{\mathcal{P}}
\newcommand{\pfinal}{\mathcal{P}^\text{final}}
\newcommand{\agents}{\mathcal{X}}
\newcommand{\q}{\mathcal{Q}}
\newcommand{\graph}{\mathcal{G}}
\newcommand{\f}{f}
\newcommand{\g}{g}
\newcommand{\opt}{\textsc{OPT}}
\newcommand{\eps}{\varepsilon}
\newcommand{\comment}[1]{}
\newenvironment{proofsketch}{%
  \proof}{\endproof}

\newcommand{\init}[1]{{\mathbf{#1}}^{\text{init}}}
\newcommand{\true}[1]{{\mathbf{#1}}^{\text{true}}}
\newcommand{\perc}[1]{{\mathbf{#1}}^{\text{perc}}}
\newcommand{\initij}[1]{{\mathbf{#1}}^{\text{init}}_i[j]}

\newcommand{\percij}[1]{{\mathbf{#1}}^{\text{perc}}_i[j]}

\newcommand{\vecj}[1]{{\mathbf{#1}}[j]}
\newcommand{\veck}[1]{{\mathbf{#1}}[k]}
\newcommand{\vecone}[1]{{\mathbf{#1}}[1]}
\newcommand{\vectwo}[1]{{\mathbf{#1}}[2]}
\newcommand{\vecstarj}[1]{{\mathbf{#1}}^{*}[j]}
\newcommand{\vecstark}[1]{{\mathbf{#1}}^{*}[k]}

\SetCommentSty{mycommfont}

\begin{document}

\title{
On classification of strategic agents who can both game and improve}
\author{Saba Ahmadi\thanks{Toyota Technological Institute at Chicago. Email: \texttt{saba@ttic.edu}. 
Author was supported by the Simons Foundation under the Simons Collaboration on the Theory of Algorithmic Fairness, and the National Science Foundation grant CCF-1733556. Part of the work was done when the author was visiting Northwestern University.
} \and%
Hedyeh Beyhaghi\thanks{Carnegie Mellon University. Email: \texttt{hedyeh@cmu.edu}. This work was done while the author was a Postdoctoral Researcher at Toyota Technological Institute at Chicago.} \and%
Avrim Blum\thanks{Toyota Technological Institute at Chicago. Email: \texttt{avrim@ttic.edu}. This work was supported in part by the National Science Foundation under grants CCF-1815011 and CCF-1733556, and the Simons Foundation under the Simons Collaboration on the Theory of Algorithmic Fairness.} \and%
Keziah Naggita\thanks{Toyota Technological Institute at Chicago. Email: \texttt{knaggita@ttic.edu}. This work was supported in part by the National Science Foundation under grant CCF-1815011, and the Simons Foundation under the Simons Collaboration on the Theory of Algorithmic Fairness.}}
\date{}
\maketitle

\begin{abstract}
In this work, we consider classification of agents who can both game and improve.  For example, people wishing to get a loan may be able to take some actions that increase their perceived credit-worthiness and others that also increase their true credit-worthiness.  A decision-maker would like to define a classification rule with few false-positives (does not give out many bad loans) while yielding many true positives (giving out many good loans), which includes encouraging agents to improve to become true positives if possible.  We consider two models for this problem, a general discrete model and a linear model, and prove algorithmic, learning, and hardness results for each.

For the general discrete model, we give an efficient algorithm for the problem of maximizing the number of true positives subject to no false positives, and show how to extend this to a partial-information learning setting.  We also show hardness for the problem of maximizing the number of true positives subject to a nonzero bound on the number of false positives, and that this hardness holds even for a finite-point version of our linear model. We also show that maximizing the number of true positives subject to no false positive is NP-hard in our full linear model.  We additionally provide an algorithm that determines whether there exists a linear classifier that classifies all agents accurately and causes all improvable agents to become qualified, and give additional results for low-dimensional data.

\end{abstract}

\section{Introduction}\label{sec:intro}

Consider a bank offering loans.  Based on observable information about applicants, it must decide which of them are loan-worthy and which are not.  For example, it might compute a credit score based on some (perhaps linear) function of observable features and then compare the result to a cutoff value.  So far, this looks like a standard binary classification problem. However, there is an additional wrinkle: individuals have agency and may be able to modify their observable features somewhat if it will help them get approved for a loan.  This wrinkle brings both challenges and opportunities.  A challenge is that some of these actions may involve ``gaming'' the system: performing activities that do not affect their true loan-worthiness such as changing how they spend on different credit cards.  An opportunity is that other actions, such as taking a money-management course, may truly help them become more loan-worthy, increasing the number of good loans the bank can give out. How can the bank best set its loan criteria in such settings to maximize the number of loans given out subject to not giving loans to unqualified applicants?  

Or, consider a school that would like to prepare students for the workforce.  There are many different career paths a student might take, so the school would like to have multiple different criteria for graduation (multiple tracks or majors) such that satisfying any one of them will earn the student a diploma.  Imagine there is a limited set of options the school can choose from, and once the school chooses some subset of them as criteria, every student selects the easiest of those criteria to fulfill (or none, if all are too hard) and then may or may not become truly qualified for the workforce, depending perhaps on the extent to which satisfying that criterion involved gaming versus true improvement.  How can the school best select criteria to maximize the number of students who become truly qualified for the workforce while minimizing the number of diplomas given to unqualified students? 

In this work we consider algorithmic and learning-theoretic formulations of such scenarios, where a binary classification must be made in the presence of both gaming and improvement actions with a goal of maximizing true-positive predictions while keeping false-positives to a minimum.  Specifically, we consider the following two formulations (given in more detail in \Cref{sec:model}). 

\begin{description}
\item[General Discrete Model:] In this formulation, we are given a weighted, colored bipartite graph with $n$ nodes on the left representing agents, and $m$ nodes on the right representing distinct possible ways agents could be considered {\em qualified} for the prize at hand (the loan, the diploma, etc.).  For example, the nodes on the right could represent different possible definitions of ``credit-worthy'' or could represent different bundles of activities sufficient to receive a diploma.  Each edge has both a {\em weight} representing the amount of effort the agent would need to achieve the given qualification and a {\em color} blue or red indicating whether the agent would indeed be truly qualified or not (respectively) if it did so.  The goal of the classifier is to select a subset $\pfinal$ of points on the right such that if each agent in the neighborhood of $\pfinal$  takes its least-cost edge into $\pfinal$, then a large number of blue edges and very few red edges are taken (many good loans and few bad loans are given out); more specific objectives will be detailed in \Cref{sec:general}.

In the learning-theoretic version of this problem, the left-hand-side of the graph is replaced with a probability distribution ${\cal D}$ over nodes (where a node is given by its neighborhood and the weights and colors of its edges).  We have sampling access to ${\cal D}$ and our goal is to find a subset $\pfinal$ of points on the right-hand-side with good performance under ${\cal D}$.  In a partial-information version, when we sample a point from ${\cal D}$ we do not get to observe its edges, only where the agent goes to and whether it was qualified.  That is, learning proceeds in rounds, where in each round we choose a subset $\points'$ of points on the right, and then for a random draw $x \sim {\cal D}$ we observe what point $p \in \points'$ (if any) was selected and the color of the edge taken.

\item[Linear Model:] In this formulation, we assume agents are points $\vec{x}\in\mathbb{R}^d$ (they have $d$ real-valued features) and there is a linear separator $f^*: \vec{a}^*\vec{x} \geq b^*$ with non-negative weights that separates the truly qualified individuals from the unqualified ones.  Agents have the ability to increase their $j$th feature at cost $\vecj{c}$ (decreasing is free) and receive value 1 for being classified as positive.  However, only some features correspond to true improvement and others involve just gaming. That is, if an agent begins at $\init{x}$ and moves to a point $\perc{x}$, their true qualification is not $f^*(\perc{x})$ but rather $f^*(\true{x})$, where $\true{x}$ agrees with $\init{x}$ in the gaming directions and with $\perc{x}$ in the improvement directions.  Movement costs and which features are improvement versus gaming  are assumed to be the same for all agents. The goal is to find a classifier that produces a large number of true positives and few false positives. Note that using $f^*$ itself will be optimal if the coordinate $j$ maximizing $\vecstarj{a}/\vecj{c}$ (having the most ``bang per buck'') is an improvement direction, so the interesting case is when this is a gaming direction.  Also note that shifting $f^*$ in this direction (adding  $\vecstarj{a}/\vecj{c}$ to $b^*$) will be a perfect classifier but may not be optimal because it does not take advantage of the ability to encourage agents to improve.   We consider settings where (a) the mechanism designer must use a linear classifier, (b) arbitrary classifiers are allowed, and (c) a polynomial-sized set $\points$ of ``target points'' is given and the mechanism designer must select some subset $\pfinal\subseteq \points$ as its classifier --- this is a special case of our General Discrete Model. 
\end{description}

In this work, we consider both models.  We give an efficient algorithm for the general discrete model for the problem of maximizing the number of blue edges taken subject to no red edges taken (maximizing the number of good loans given out subject to no bad loans) and show how to extend this to the partial-information learning setting.  We also show hardness for the problem of maximizing the number of blue edges subject to a nonzero bound on the number of red edges, and show that this hardness holds even for the simplest finite-point linear model. Furthermore, we show the problem of maximizing the number of true positives subject to no false positives is NP-hard in the linear model when we are not given a polynomial-sized set of target points.  We additionally give algorithms for the linear model.
We provide an algorithm that determines whether there exists a linear classifier which classifies all agents accurately and causes all improvable agents to become qualified. In the special two-dimensional case, we design a linear classifier maximizing the number of true positives minus false positives; and a general (not necessarily linear) classifier that maximizes true positives subject to no false positives.

\subsection{Related Work}

There is an exciting and growing literature on decision-making in the presence of strategic agents. Much of this work considers agents whose actions are only gaming and do not change their true label (see \cite{Hardt2016,revealed_preferences,Hu:2019:,Milli2018TheSC, Ahmadi2021TheSP,adversarial_games_pred,Frankel2019ImprovingIF,Braverman2020TheRO} among others)  but researchers have also been investigating mechanism design in the presence of agents who can both game and improve {\cite{Kleinberg2018HowDC, harris2021stateful, Alon2020MultiagentEM, xiao2020optimal, Miller2019StrategicCI, Haghtalab2020MaximizingWW, Bechavod2020CausalFD, Shavit2020LearningFS}}.  

\citet{Kleinberg2018HowDC} consider a single agent with a variety of gaming and improvement actions available, that are then converted into observable features through an effort-conversion matrix.  They then examine mechanisms for incentivizing desired action vectors, showing among other things that any vector that can be incentivized by a monotone mechanism can also be incentivized by a linear mechanism.  
\citet{harris2021stateful} consider a multi-round version of the \citet{Kleinberg2018HowDC} model in which true improvements carry over to future rounds whereas gaming effort do not; they show that in this model, the principal (the decision-maker) can incentivize the agent to produce a greater range of desirable behaviors. 

\citet{Alon2020MultiagentEM} consider a multi-agent extension of the \citet{Kleinberg2018HowDC} model, where agents all begin at the same place (the origin) but each have their own effort-conversion matrix. The goal of the designer is to choose an evaluation mechanism---mapping observable features to payoffs---that encourages all agents to take {\em admissible} actions, assuming that agents will maximize payoff subject to budget constraints.  They specifically consider the case (1) that there is a single admissible action vector, and (2) that individual actions are either improvement or gaming actions and no agent should take a gaming action. Among other results they show that unlike in \cite{Kleinberg2018HowDC}, nonlinear evaluation mechanisms can now be more powerful than linear ones; they also analyze the complexity of a variety of associated optimization problems.  We can think of our setting to some extent in this language by viewing any action that makes an agent truly qualified as ``admissible'' (and specifically the blue edges in our general discrete model).  However, two key distinctions are (1) in our setting we can only give the loan/diploma or not---we do not have the flexibility to choose arbitrary payoffs, and (2) we assume agents may begin at different starting locations (but have the same costs for movement in our linear model). 

\citet{xiao2020optimal} define a problem they call the {\em Multiple Agents Contract Problem} which is very similar to our General Discrete Model, except instead of binary (red/blue) colors, the edges have different values to the principal, and instead of producing a classification, the principal can assign an arbitrary payment profile to the right-hand-side nodes.  They prove that maximizing payoff to the principal is NP-hard, and give an algorithm for a case of related agents in which there is a certain strict ordering among agents and costs. 

\citet{Shavit2020LearningFS}, building on \citet{Miller2019StrategicCI}, consider the goal of getting agents to improve without loss of predictive accuracy.  As in our setting, they assume agents begin a different starting locations, and then modify their profiles from there, and they also consider a learning formulation.  However, their focus is on a regression model in which agents' payoffs are an inner product of their observable features with a decision vector; this means that the incentives are basically the same no matter what the initial location of an agent is.  In contrast, in our binary classification setting, even in the linear model the effect of a proposed classifier on an agent may depend greatly (and in a non-convex manner) on the initial location of the agent.  \citet{Bechavod2020CausalFD} also consider a linear regression learning setting: agents arrive one at a time iid from a fixed distribution and then modify their state by changing a single variable based on the current regression vector.  As in our linear model, some directions are improvement and some are gaming.  They consider a limited feedback setting where the learner sees only the dot-product of the agent's true position with the true regression function, plus noise, and the learner's goal is to recover the true regression function.

\citet{Haghtalab2020MaximizingWW} consider a similar setting to ours in which there are improvement and gaming actions, and the designer is limited to binary classification, where agents receive value 1 for being classified as positive. Among other results, they give approximation algorithms for the goal of maximizing the total amount of true improvement that occurs when the allowed mechanisms are linear separators and agents have $\ell_2$ movement costs.  In contrast, our goal is to maximize true positive classifications while minimizing false positives, and in the linear case our movement cost assumptions are {somewhat different}.

\paragraph*{Organization of the Paper.}
{\Cref{sec:model} introduces the general discrete model and linear model more formally.
In \Cref{sec:general}, we give an efficient algorithm for the problem of maximizing the number of true positives subject to no false positives in the general discrete model, and provide hardness results for the problem of maximizing the number of true positives subject to a nonzero bound on false positives (in either the general discrete model or the linear model when arbitrary classifiers are allowed) and hardness for the problem of maximizing the number of true positives subject to no false positives in the linear model when arbitrary classifiers are allowed. In \Cref{sec:learning}, we consider a learning-theoretic version of the problem of maximizing true positives subject to no false positives, and provide efficient learning algorithms as well as upper and lower bounds on the number of samples needed. In \Cref{sec:linear}, we focus on the linear model and provide algorithms specific to this setting. We provide an algorithm that determines whether there exists a linear classifier which classifies all agents accurately and causes all improvable agents to become qualified. In the special two-dimensional case, we design a linear classifier maximizing the number of true positives minus false positives; and a general (not necessarily linear) classifier that maximizes true positives subject to no false positives.}
\section{Model}\label{sec:model}

We study a binary classification problem. As the mechanism designer or classifier, we would like to maximize the number of agents we correctly classify as positive (true positives), and minimize the number of unqualified agents we misclassify as positive (false positives). 

Agents are assumed to be utility maximizers and wish to be classified as positive.  Each agent $i \in \{1, \ldots, n\}$ has a set of actions it can perform, and it will choose the cheapest of these that causes it to be classified as positive if that cost is less than its value on receiving a positive classification.  We use $\q$ to denote the set of truly qualified agents.  If an agent is initially not qualified (not in $\q$), some of its actions may cause it to become truly qualified, whereas others may not.  However, the classifier cannot see which action was taken, only the observable result of that action.  Therefore, the challenge of the mechanism designer is to determine which observable results to classify as positive to maximize  correct positive classifications while minimizing false positives.

\subsection{General Discrete Model}\label{sec:gen_model}
In this model, we assume that as a mechanism designer we are given a polynomial-sized set $\points$ of criteria we may select from (e.g., graduation criteria or criteria for being approved for a loan), and are limited to choosing some subset $\pfinal \subseteq \points$ as the criteria we will use.  We then will classify as positive any agent that meets any one of these criteria, {and as negative any agent who does not}.  Specifically, we are given a weighted, colored bipartite graph with the $n$ agents on the left and the set $\points$ of criteria on the right.  Edge $(i,j)$ corresponds to agent $i$ taking an action to satisfy criteria $j$ and is colored blue or red depending on whether that action would make the agent truly qualified or not, respectively.  Each edge also has a weight representing its cost to that agent, and only actions whose costs are less than the value to the agent of being classified as positive are shown.   Given a set  $\pfinal \subseteq \points$ chosen by the mechanism designer, each agent in the neighborhood of $\pfinal$ will choose its cheapest edge into $\pfinal$ as the action it will take, and will be classified as positive by the mechanism; agents not in the neighborhood of $\pfinal$ will be classified as negative.

We also consider a learning-theoretic version of this problem, where the left-hand-side of the graph is replaced with a probability distribution ${\cal D}$ over nodes.  We have sampling access to ${\cal D}$ and our goal is to find a subset $\pfinal$ of points on the right-hand-side with good performance under ${\cal D}$.  In a partial-information (bandit-style) version, when we sample a point from ${\cal D}$ we do not get to observe its edges, only where it goes to and whether it was qualified.  That is, learning proceeds in rounds, where in each round we choose a subset $\points'$ of points on the right, and then for a random draw $x \sim {\cal D}$ we observe what point $p \in \points'$ (if any) was selected and the color of the edge taken.

\begin{figure}[ht!]
\centering
\includegraphics[width=10cm]{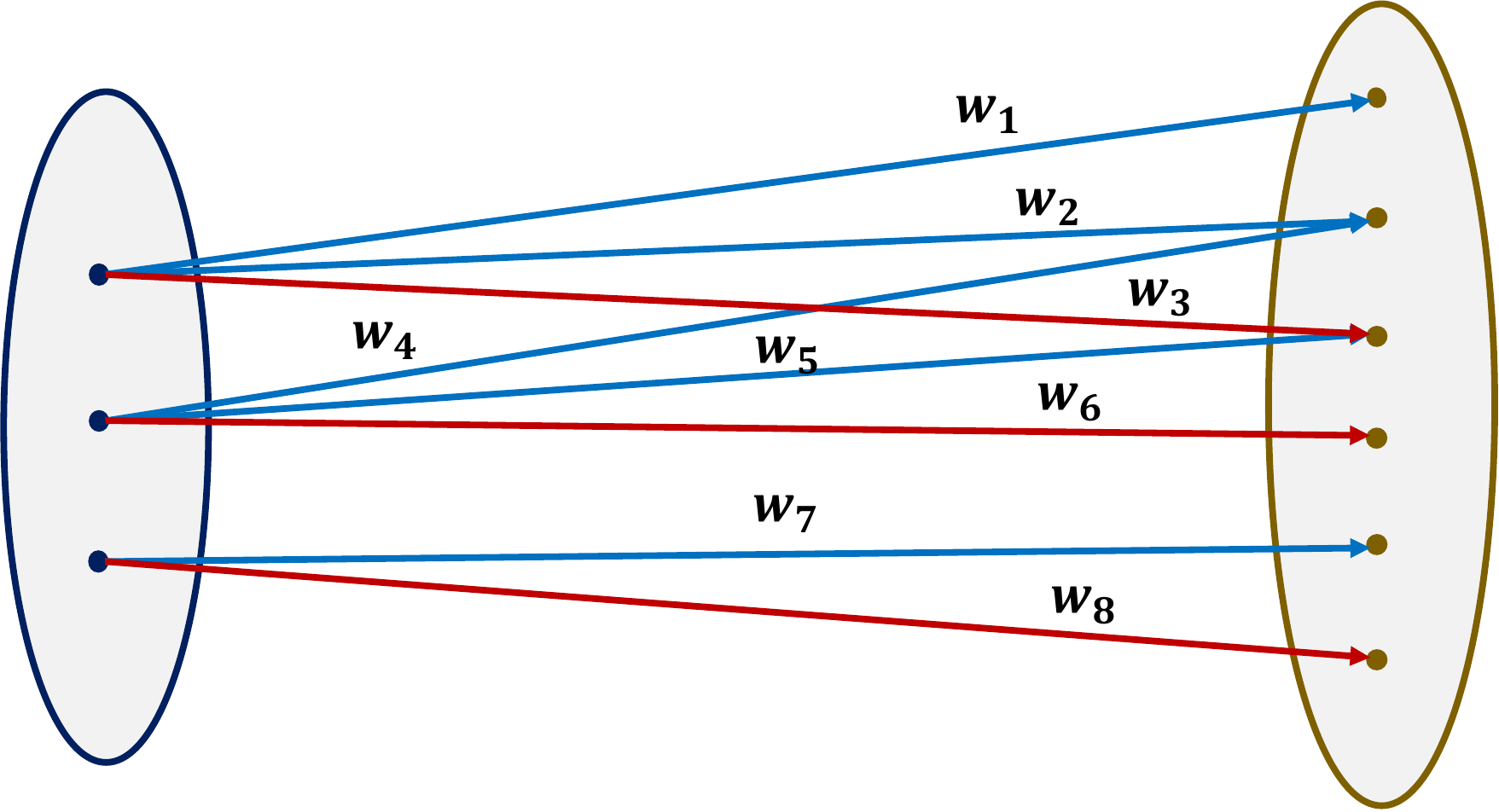}
\caption{\textit{Points on the left are the agents, and those on the right are the set $\points$ of possible criteria; $w_i$ is the cost of satisfying the criterion. A red edge means the agent taking that action would not truly be qualified. A blue edge means that the agent taking that action would be qualified.}}
\label{fig:matching}
\end{figure}

\subsection{Linear Model}\label{sec:linear_model}

In the linear model, agents have $d$ real-valued features.  Each agent $i$ begins at an initial point $\init{x}_i \in \mathbb{R}^d$, and there is assumed to be a linear threshold function $f^*: \vec{a}^*\vec{x} \geq b^*$ with non-negative weights that separates the truly qualified individuals from the unqualified ones.  Agents have the ability to increase their $j$th feature at cost $\vecj{c}$ (decreasing is free) and receive value 1 for being classified as positive.  However, only some features correspond to true improvement and others involve just gaming. That is, if an agent begins at $\init{x}$ and moves to a point $\perc{x}$, their true qualification is not $f^*(\perc{x})$ but rather $f^*(\true{x})$, where $\true{x}$ agrees with $\init{x}$ in the gaming directions and with $\perc{x}$ in the improvement directions.   On the other hand, the classification rule can only be based only on $\perc{x}$ and not $\true{x}$ (or $\init{x}$).  Movement costs and which features are improvement versus gaming  are assumed to be the same for all agents.  So, for any agent $i$, $cost(\init{x}_i,\perc{x}_i)=\sum_{j=1}^d \vecj{c} \left(\percij{x}-\initij{x}\right)^+$, where $x^+=\max\{x,0\}$ and $\vecj{c}$ is the cost per unit of movement in the positive direction of dimension $j$. 

We consider settings where (a) the mechanism designer must use a linear classifier (a linear threshold function), (b) arbitrary classifiers are allowed, and (c) a polynomial-sized set $\points$ of ``target points'' is given and the mechanism designer must select some subset $\pfinal\subseteq \points$ as its classifier.  Notice that this last case is a special case of the general discrete model because given each initial state $\init{x}_i$, we can compute the costs to move to each $p \in \points$ and whether doing so will make the agent truly qualified, to produce the desired weighted, colored bipartite graph. 

\begin{figure}[ht!]
\centering
\includegraphics[width=11cm]{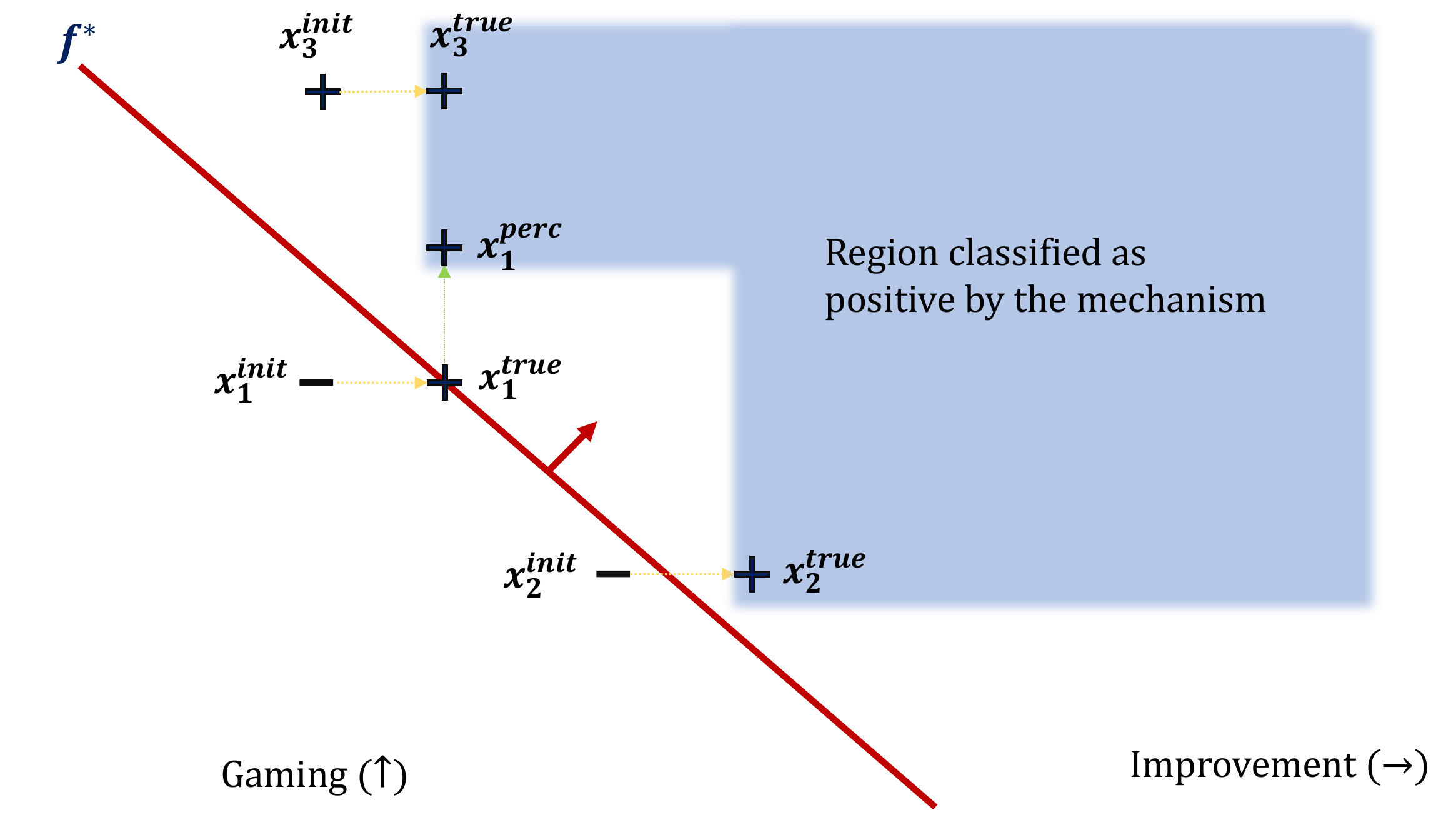}
\caption{\it An example of the linear model (the horizontal axis is an improvement direction and the vertical axis is a gaming direction) with a mechanism using a non-linear classifier.  There are three agents, two of whom are initially not qualified.  All three become qualified and are correctly classified as positive by the mechanism.}
\label{fig:matching}
\end{figure}

\section{Algorithmic and Hardness Results}
\label{sec:general}

In this section we first provide an algorithm for the problem of maximizing the number of true positives subject to no false positives in the general discrete model. Then, we provide hardness results for the problem of maximizing the number of true positives subject to a nonzero bound on false positives (in either the general discrete model or the linear model when arbitrary classifiers are allowed) and hardness for the problem of maximizing the number of true positives subject to no false positives in the linear model when arbitrary classifiers are allowed. Later in \Cref{sec:learning} we extend our algorithmic results to the learning model and in \Cref{sec:linear} we give algorithms for learning linear classifiers in the linear model.

\subsection{Maximize True Positives Subject to No False Positives}

The main result of this section is an algorithm that given a weighted, colored bipartite graph $\graph$ with agents, $\agents$, on the left and potential criteria, $\points$, on the right, finds $\pfinal \subseteq \points$ such that using $\pfinal$ as the criteria  maximizes the number of agents taking a blue edge (true positive) subject to no agent taking a red edge (false positive). We call the agents that take a blue edge \emph{improving agents} and the agents taking a red edge \emph{gaming agents}. The algorithm, although simple in structure, satisfies strong properties noted afterwards; and serves as the building block of the learning algorithms in \Cref{sec:learning}. Furthermore, as shown in the following subsection, natural generalizations of the objective function make the problem computationally hard. Therefore, the algorithm together with the hardness results tightly characterize the settings for which there is an efficient algorithm, or the problem is NP-hard.

\textbf{Overview of \Cref{alg:maximal}}. The algorithm takes in a weighted, colored bipartite graph $\graph = (\agents \cup \points,E)$ and outputs  $\pfinal$, a subset of $\points$ that specifies the final criteria. Initially, $\pfinal$ is set to $\points$. The algorithm proceeds in rounds. In each round, it visits all the nodes (agents) in $\agents$ to determine whether there is an agent who takes a red edge to its lowest cost neighbor $p \in \pfinal$. If there is such a gaming agent, its corresponding criteria, $p$, is removed from $\pfinal$. These rounds continue until there is no gaming agent and therefore no removal of criteria in a single round, or the current set of criteria is empty. 

\begin{algorithm}
    \SetNoFillComment
    \SetAlgoLined
    \SetKwInOut{Input}{Input}
    \SetKwInOut{Output}{Output}
    \SetKw{Return}{return}
     \DontPrintSemicolon
    \Input{A bipartite graph  $\graph = (\agents \cup \points,            E)$ with edge weights $w_{e}$. Outgoing edges assumed sorted by weight.
           Red edges  $E_{R} \subseteq E$.
           Blue edges $E_{B} \subseteq E$.
        }
        
    \Output{$\pfinal$ }
    $\pfinal \leftarrow \points$\ \tcp{Initialization of the set}

    \While{$\pfinal \neq \emptyset$}{
        $flag = 0$\;
        \tcc{Loop through all $x_i \in \agents$}
        \For{$i=1, 2, \cdots$}  { 
            Let $e = (x_i, p \in \pfinal)$ be the outgoing edge from $x_i$ with lowest weight\;
            \If{$e\in E_R$}{
                $flag = 1$\ \tcp{at least one agent is gaming}
                $\pfinal \leftarrow \pfinal \setminus \{p\}$\
            }
        }
        \If{flag is $0$}{
            \Return $\pfinal$\
        }
    }
    \Return $\emptyset$\ \tcp{When $0$ false positive is not possible}
    \caption{Maximize true positives subject to no false positives.}
    \label{alg:maximal}
\end{algorithm}

\proposition{\label{prop:running-time-greedy-alg}\Cref{alg:maximal} has running time of $O(|\points|n)$}. 
\proof Proof in~\Cref{app:hardness}.

\begin{theorem}
\label{proof:greedy-correctness}
\Cref{alg:maximal} finds the set of criteria, $\pfinal$, that maximizes {the number of true positives subject to no false positive.} \end{theorem}
\begin{proof}
Let A be the improving agents (agents taking blue edges) associated with the set of criteria $\pfinal$. We show that having any other set $Q \subseteq \points$ as the criteria, either causes an agent to take a red edge, or no more than $|A|$ agents to take blue edges. To do so, consider partitioning $Q$ into two subsets $Q^F$ and $Q^{\bar{F}}$, where $Q^F \subseteq \pfinal$ and $Q^{\bar{F}} \subseteq \points \setminus \pfinal$.

First, we show that if $Q^{\bar{F}} \neq \emptyset$, an agent takes a red edge. To prove this claim, suppose by  contradiction that $Q^{\bar{F}}$ is nonempty and consider the first time the algorithm deletes an element $p \in Q^{\bar{F}}$. At this stage, the set of criteria in the algorithm $\points'$ is a superset of $Q^{\bar{F}} \cup \pfinal$. By definition, $p$ is the lowest-weight neighbor of a gaming agent, $a$, in $\points'$. This implies that $p$ is also the lowest-weight neighbor of $a$ in $Q \subseteq Q^{\bar{F}} \cup \pfinal \subseteq \points'$, and $a$ is a gaming agent given the criteria set $Q$. This implies the claim. 

Secondly, we show that among the sets of criteria with no gaming agent, $\pfinal$ has the highest number of improving agents. The previous claim implies that any set of criteria with no gaming agent is a subset of $\pfinal$. Now, we need to show that among $Q \subseteq \pfinal$, $\pfinal$ has the largest set of improving agents. This is trivial, since by considering a subset we may only lose on agents in $A$ that do not have a neighbor in $Q$ or their lowest-weight edge is red. Therefore, any $Q \subseteq \pfinal$  has at most $|A|$ improving agents.
\end{proof}

\Cref{alg:maximal} satisfies the following strong properties.
\begin{enumerate}[label=(\alph*)]
    \item \label{pr:point-wise-optimal} {\em point-wise optimality}:
    For any agent $i$, if there exists a solution in which $i$ takes a blue edge and no agent takes a red edge, then the algorithm finds such a solution. 
    \item {\em general for weighted setting}: The algorithm works optimally in the more general setting that each agent has a weight and the objective is to maximize the sum of weights of improving agents subject to the constraint of no gaming agent. This is a direct implication of property \ref{pr:point-wise-optimal}.
    \item {\em max-min fairness:} Suppose the agents are from different populations and the objective is to maximize the minimum number of agents improving from each population subject to no gaming. By property \ref{pr:point-wise-optimal}, the algorithm satisfies this max-min fairness notion.
    \item {\em heterogeneous utilities}: The algorithm works optimally in the more general setting that agents have different values for being classified positive.
    \item {\em minimizing the total cost of improvement}: Since the algorithm only removes $p \in \points$ that causes an agent to game, with $\pfinal$ each agent incurs the minimal cost subject to no agent gaming.
\end{enumerate}

\begin{remark} 
The sets of criteria satisfying the no false positive constraint is not downward closed. In other words, a subset of a set of criteria that satisfies the no false positives property does not necessarily satisfy this property.
\end{remark}

\subsection{Hardness Results}

In this part, we prove hardness results for maximizing the number of true positives when the constraints in the previous subsection are relaxed. First, we show that if we relax the no false positives constraint to a bounded number of false positives, the problem becomes NP-hard; moreover, this holds even for the simpler linear model. Then, for the linear model, we show if we are not given a finite set of potential criteria $\points$, it is NP-hard to find criteria that maximize true positives subject to no false positives.

\begin{theorem}\label{thm:atmost_k_game}
Given the initial feature vectors of agents $\init{x}_1, \init{x}_2, \ldots, \init{x}_n \in \mathbb{R}^d$ and a set $\points$ of potential criteria, the problem of finding a subset $\pfinal \subseteq \points$ that maximizes the number of true positives subject to at most $k$ false positives is NP-hard. 
\end{theorem}

\begin{proofsketch} The proof is done by a reduction from the Max-$k$-Cover problem with $n$ elements where the goal is to choose $k$ sets covering the most elements. For every element $e_i$ in the Max-$k$-Cover, we consider agent $i$, and for every set $S_j$ in the Max-$k$-Cover problem we consider agent $n+j$ and a target point $\vec{p}_j$. The coordinates of the initial points and the target points are set such that agent $i$ corresponding to element $e_i$ can only move to target point $\vec{p}_j$ such that $e_i \in S_j$ and become a true positive; moreover, agent $n+j$ corresponding to set $S_j$ can only move to target point $\vec{p}_j$ and become a false positive. On the one hand, since including each $\vec{p}_j$ in the final set of criteria, $\pfinal$, causes exactly one agent to be a false positive, $\pfinal$ must contain at most $k$ target points. On the other hand, to maximize the number of true positives a set of $k$ target points that the maximum number of agents can reach to it must be selected. This is equivalent to the Max-$k$-Cover solution.
A formal proof is included in \Cref{app:hardness}.
\end{proofsketch}

\begin{theorem}\label{thm:relaxed-no-destination-pts}
Suppose we are given a set of $n$ agents where $\init{x}_1, \init{x}_2, \ldots, \init{x}_n$ denote their initial feature vectors. Deciding whether there exists a set of target points $\pfinal\subseteq \mathbb{R}^d$ for which all the agents become true positives is NP-hard. 
\end{theorem}

\begin{proofsketch}

The proof is done by a reduction from the approximate version of the hitting set problem where given a set  of elements, $\mathcal{E}=\{e_1, \ldots, e_n\}$ and  a family of sets of elements, $\mathcal{F}=\{S_1, S_2, \ldots, S_m\}$, the goal is to find a minimum size set $S^*$ that intersects all $S_i$. We construct an $n+1$-dimensional space, where the first $n$ dimensions are improvement dimensions and correspond to the $n$ elements, and the last dimension is gaming. We consider two sets of agents. For each $S_i$, we consider a corresponding agent $i$; these are the \emph{usual} agents. We also consider agent $m+1$, a \emph{special} agent that does not correspond to any particular set. The construction is such that each agent needs to move $2k$ units along the improvement dimensions to become truly qualified. Further details of the construction can be found in the full proof.
The proof includes two directions. (1) If all the agents can become true positives by reaching to a set of target points $\pfinal\subseteq \mathbb{R}^d$, then we can construct a hitting set of size at most $2k$; and (2) if it is not possible, then there does not exist a hitting set of size $k$.

We briefly cover the key ideas in each direction. To show the first direction, suppose all the agents can become true positives when presented with target points $\pfinal\subseteq \mathbb{R}^d$. Consider the target point that each agent selects. Using our construction, we show the special agent does not afford to reach to the target points of the usual agents. Also, for each usual agent $i$, there exists element $e_j$ in their corresponding set such that the target point of the special agent has value more than $1$ in coordinate $j$. In order for the special agent to afford to reach to its target point, the number of improvement coordinates with value at least $1$ must be at most $2k$. The elements corresponding to these coordinates constitute a hitting set of size at most $2k$.
To prove the reverse direction we argue: if there exists a hitting set $S^*$ of size $k$, there is a set of target points that encourages all the agents to become true positives. To do so, we construct a set of target points $\pfinal = \{\vec{p}_1, \ldots, \vec{p}_{m+1}\}$, using the elements in the hitting set, that when the size of the hitting set is $k$ makes every agent become true positive. A formal proof is included in \Cref{app:hardness}.
\end{proofsketch}

The following is a direct corollary of \Cref{thm:relaxed-no-destination-pts}.

\begin{corollary}
Given the initial feature vectors of agents, $\init{x}_1, \init{x}_2, \ldots, \init{x}_n \in \mathbb{R}^d$, finding a set of target points $\pfinal \subseteq \mathbb{R}^d$ that maximizes the number of true positives subject to no false positives is NP-hard.
\label{cor:d-dim-max-TP-no-FP}
\end{corollary}

\section{Learning Results}
\label{sec:learning}

In this section we consider a learning-theoretic version of our problem, where the left-hand-side of the graph is replaced with a probability distribution ${\cal D}$ over nodes.  We have sampling access to ${\cal D}$ and our goal is to find a subset $\pfinal$ of points on the right-hand-side with good performance under ${\cal D}$. 
We provide two different algorithmic results and upper bounds on the number of samples for producing a good solution, depending on the information each sample reveals. The first upper bound works for the case where by sampling an agent, its neighborhood (neighboring edges, their colors and weights) is revealed. The second upper bound works in a partial-information (bandit-style) setting, where when we sample a point from ${\cal D}$ we do not get to observe its edges, only where it goes to and whether it was qualified. Finally, we provide a lower bound on the necessary number of samples for any algorithm. The lower bound holds even for the simpler linear model. On the technical side, the algorithms in \Cref{sec:learning_full,sec:learning_partial} use \Cref{alg:maximal} as a subroutine and generalize it to a broader setting. The lower bound in \Cref{sec:learning_lower}, however, holds for \emph{any} PAC learning algorithm, and requires substantially different ideas.

The following definition is crucial in this section.
\begin{definition}[$\opt$, performance, and error]
Let $\opt$ be the maximum probability mass of true positives achievable subject to zero false positives. We denote the probability mass of true positives of an algorithm as its \emph{performance} and the probability mass of false positives as its \emph{error}. A hypothesis is desired if it has comparable performance to $\opt$ and small error.
\end{definition}

\subsection{Sufficient Number of Samples in the Full Information Setting}\label{sec:learning_full}

The main result of this section is that a number of samples linear in $|\points|$ and $1/\eps$ is sufficient for \Cref{alg:maximal} to learn a desired hypothesis with high probability.
Specifically, suppose the learner has access to a weighted, colored bipartite graph $\graph = (\agents \cup \points,E)$, where $\agents$ are sampled from ${\cal D}$, and $\points$ is the set of the potential criteria. The learner runs \Cref{alg:maximal} with the graph as the input and uses the algorithm output, $\pfinal \subseteq \points$, as its hypothesis, 
i.e., after the training phase it classifies any agent with an edge to $\pfinal$ as positive and any other agent as negative.  We show that a linear number of samples is sufficient so that with high probability, the probability mass of true positives classified by $\pfinal$ is close to $\opt$ and the probability mass of false positives is small.

\begin{theorem}
\label{thm:distributional-upper-bound}
Consider $\pfinal$ as the outcome of \Cref{alg:maximal} on $\graph = (\agents \cup \points,E)$, where $\agents$ contains samples from ${\cal D}$. For any $0<\eps, \delta \leq 1$, if $|\agents| \geq \eps^{-1} (\ln(2)|\points|+\ln(1/\delta))$, then with probability at least $1-\delta$ 
the set $\pfinal$ achieves performance at least 
$\opt-\eps$ (i.e., at least $\opt-\eps$ probability mass of true positives)
subject to at most $\eps$ error ($\eps$ probability mass of false positives).
\end{theorem}

\begin{proof}
First, we prove the error bound. By \Cref{alg:maximal}, $\pfinal$ has error $0$ for $\agents$. Now, consider a set of criteria $\points' \subseteq \points$ with error at least $\eps$ for distribution $D$. The probability of the error being equal to $0$ over $\agents$ is at most $(1-\eps)^{|\agents|}$. The number of subsets of $\points$ is $2^{|\points|}$. So, by union bound the probability that there exists a set of criteria $\points'$ with error greater than $\eps$ over $D$ and error $0$ over $\agents$ is at most $2^{|\points|}(1-\eps)^{|\agents|}\leq 2^{|\points|}e^{-\eps|\agents|}$. This probability is at most $\delta$ for $|\agents| = \eps^{-1}\Big(\ln(2)|\points| + \ln(1/\delta)\Big)$. Therefore, with  $\eps^{-1}\Big(\ln(2)|\points| + \ln(1/\delta)\Big)$ number of samples, with probability at least $1-\delta$, $\pfinal$ has at most $\eps$ error.

Next, we show the performance guarantee. Recall that $\opt$ is the maximum achievable probability mass of true positives subject to no false positive for distribution $D$. Set $\agents$ contains a subset of the points in the distribution. \Cref{alg:maximal} only deletes points from $\points$ if they cause a false positive in $\agents$. Therefore, the output of the algorithm on agents $\agents$, $\pfinal$, is a superset of the criteria set of the optimal zero-error solution. 
This means that the probability mass of examples predicted positive is at least $\opt$, and since at most an $\eps$ probability mass is false-positives, the performance of $\pfinal$ is at least $\opt-\eps$.
\end{proof}

\subsection{Sufficient Number of Samples in the Partial Information Setting}\label{sec:learning_partial}

In this section we consider a partial information (bandit-style) setting. Similar to before, the learner has access to a sample set $\agents$ drawn from $D$ and a set of potential criteria $\points$. However, observing a sample in $\agents$ does not reveal its edges, and the learner can only observe the criterion that the sample selects and whether it becomes truly qualified. The main result of this section is an algorithm,
 \Cref{alg:upper-bound-no-input-graph}, for this setting and a guarantee on the number of samples sufficient for it to achieve performance at least $\opt-\eps$ and error at most $\eps$ with high probability.

\begin{algorithm}
    \SetNoFillComment
    \SetAlgoLined
    \SetKwInOut{Input}{Input}
    \SetKwInOut{Output}{Output}
    \SetKw{Continue}{continue}
    \SetKw{Return}{return}
    \Input{$\mathcal{P}$}
    \Output{$\mathcal{P}^{final}$}
    $\points^\text{final} \leftarrow \points$\;
    \While{$\points^\text{final}\neq \emptyset$}{
        Sample $\mathcal{X}\sim\mathcal{D}$ of size $\frac{1}{\eps}\ln {\frac{|\points|}{\delta}}$ \;
        \If{$\exists x \in \mathcal{X}$ such that $x$ takes a red edge to $p\in \points^\text{final}$}{
            $\points^\text{final} \leftarrow \points^\text{final} \setminus \{p\}$\;
            \Continue\;
        }
        \tcc{if no one from $\mathcal{X}$ takes a red edge:}
        \Return $\points^\text{final}$\;
    }
    \Return $\emptyset$\;
\caption{Learning a high performance low error $\pfinal$ in partial-information setting}
\label{alg:upper-bound-no-input-graph}
\end{algorithm}

\textbf{Overview of \Cref{alg:upper-bound-no-input-graph}}. In each iteration, a set of examples of size $\eps^{-1}\ln (|\points|/\delta)$ is sampled. After agents select points in $\points$ (if any), we observe the points selected and whether they became truly qualified (in a real-world application, one can think of performing a test to check if each agent is truly qualified). If some agent does not become truly qualified (fails the test), the algorithm deletes the point they have selected. If a set $\points^\text{final}$, survives for $\eps^{-1}\ln ({|\points|}/{\delta})$ subsequent examples, the algorithm terminates and returns $\pfinal$ as the the final set of criteria of the algorithm. Since the number of false positives (agents taking red edges) is bounded by $|\points|$, the algorithm will terminate after at most $\eps^{-1}|\points|\ln (|\points|/\delta)$ samples.

The following theorem proves that with a high probability, \Cref{alg:upper-bound-no-input-graph} outputs $\pfinal$ with a high performance and a low error.

\begin{theorem}
For any $0<\eps, \delta \leq 1$, \Cref{alg:upper-bound-no-input-graph} by using at most $\eps^{-1}|\points|\ln (|\points|/\delta)$ total samples outputs a set of criteria $\pfinal$ that with probability at least $1-\delta$ achieves performance at least $\opt-\eps$ (i.e., at least $\opt-\eps$ probability mass of true positives) subject to at most $\eps$ error ($\eps$ probability mass of false positives).
\end{theorem}

\begin{proof}

First, we prove the error bound. Consider the sequence of $\pfinal$ at the beginning of each iteration of the while loop in \Cref{alg:upper-bound-no-input-graph}. Let these sets be $\pfinal_1, \pfinal_2, \ldots$. The probability that $\pfinal_i$ with error greater than $\eps$ over $D$ does not produce any false positives in the following $\eps^{-1}\ln (|\points|/\delta)$ samples is at most $(1-\eps)^{\eps^{-1}\ln (|\points|/\delta)}$. Note that the while loop in \Cref{alg:upper-bound-no-input-graph} runs at most $|\points|$ times. Therefore, the number of distinct sets of $\pfinal_i$ considered in the algorithm is at most $|\points|$. By a union bound, the probability that there exists $\pfinal_i$ with error greater than $\eps$ over $D$ and error $0$ over the samples in its iteration is at most $|\points|(1-\eps)^{\eps^{-1}\ln (|\points|/\delta)}$, which using $1-x \leq e^{-x}$ is at most $\delta$.

Proving the performance guarantee is identical to that of \Cref{thm:distributional-upper-bound}.
\end{proof}

\subsection{Necessary Number of Samples}\label{sec:learning_lower}

The main result of this section is a lower bound on the necessary number of samples for learning a desired hypothesis. The lower bound provided holds even for the simpler linear model. To restate the setup, suppose  the learner has access to a set of initial positions of agents $\agents$ and a set of potential criteria (also called target points in the linear model) $\points$ where $\agents$ are sampled from distribution ${\cal D}$. We lower bound the required number of samples for any learning algorithm that with probability at least $1/2$ achieves high performance and low error.

\begin{theorem} 
Any algorithm for PAC learning a set $\pfinal$ that with probability at least $1/2$ achieves performance at least $(3/4)\cdot \opt$ (i.e., at least $(3/4)\cdot \opt$ probability mass of true positives) subject to at most $\eps$ error ($\eps$ probability mass of false positives) must use $\Omega(|\points|/\eps)$ examples in the worst case.
\end{theorem}

For ease of notation, in the proof we use $m := |\points|$. 

\begin{proof}
We construct a concept class, a distribution of agents and a set of potential criteria (target points) that forces any PAC learning algorithm to take many samples. First, consider a concept class $C$ that includes solutions whose final set of target points $\pfinal$ has size exactly $3m/4$. Therefore,  $|C| = \binom{m}{3m/4}$. Target concept $c$ is chosen randomly from $C$. Secondly, we construct a linear setting. We consider a two-dimensional space where $\f^*: \vec{x}[1]+\vec{x}[2] \geq 2m$. Let $\points=\{\vec{p}_1, \ldots, \vec{p}_m\}$. All $\vec{p}_i$ satisfy $\vec{p}_i[1]+\vec{p}_i[2] = 2m$ and $\vec{p}_i[1] = 2i$. The costs of moving in either dimension is $1$ per unit of movement, i.e., $\vec{c}[1]=\vec{c}[2]=1$. Dimension $1$ is an improvement dimension and dimension $2$ is a gaming dimension. Let the set of examples $S=\{\vec{x}_1, \vec{x}_2, \ldots, \vec{x}_{2m}\}$ denote the distinct potential initial positions of any agents. We construct the examples in proximity of the target points such that each example can afford to move to exactly one target point, called its \emph{designated} target point. More formally, for $i \leq m$, let $\vec{x}_i[1] = \vec{p}_i[1]-1, \vec{x}_i[2] = \vec{p}_i[2]$ and $\vec{x}_{m+i}[1] = \vec{p}_i[1], \vec{x}_{m+i}[2] = \vec{p}_i[2]-1$. With this setup, examples $i$ and $m+i$ are in proximity of their designated target point $\vec{p}_i$. Examples $\vec{x}_i$ such that $i \leq m$ are \emph{improving examples} since any agent with initial position $\vec{x}_i$ becomes truly qualified by moving to their designated target points and examples $\vec{x}_i$ such that $i > m$ are \emph{gaming examples} since any agent with initial position $\vec{x}_i$ does not become truly qualified. Finally, we consider a distribution $D_c$ over the examples. Let $P_G$ be the target points not included in the concept $c$. For each $i$ such that $\vec{p}_i \in P_G$, $\Pr[\vec{x}_{m+i}] = 128\eps/m$;  for each $i$ such that $\vec{p}_i \notin P_G$, $\Pr[\vec{x}_{m+i}] = 0$; and for each $i$, $\Pr[\vec{x}_{i}] = (1-32\eps)/m$.  With this probability distribution, there is a $0$ probability mass over gaming examples with designated target point $\notin P_G$, total probability mass of $32\eps$ distributed uniformly over gaming examples with designated target points $\in P_G$, and total probability mass of $1-32\eps$ distributed uniformly over improving examples. Note that with this construction, the target concept $c$ uses $\pfinal = \points \setminus P_G$ which achieves performance $\opt = 3/4$ and error equal to $0$.

Now, let $L$ be any PAC learning algorithm for $C$. Consider running $L$ when the target concept $c \in C$ is chosen randomly and the input distribution is $D_c$. Recall that $P_G$ is the set of target points not included in $c$ which is also the set of target points that positive mass of gaming examples can reach to. The purpose of the algorithm is to learn $P_G$. For this purpose the algorithm only benefits from sampling gaming examples. This is because each gaming example $\vec{x}_{m+i}$ reveals $\vec{p}_i \in P_G$; and in contrast, since improving examples are distributed uniformly across \emph{all} target points $\points$, sampling an improving example does not provide any information about $P_G$. We denote observing a gaming example with designated target point $\vec{p} \in  P_G$ as ``observing a point in $P_G$''. We consider set $O_G$ which includes any observed point in $P_G$. In what follows we assume $L$ never includes any point in $O_G$ in its set of final target points $\pfinal$. This is because any point in $O_G$ causes an error of $128\eps/m$, while replacing it with any $\points \setminus O_G$ causes less error while the probability mass of true positives does not decrease.

The proof consists of two parts. In the first part, we argue that the number of $\vec{p} \in  P_G$ that the algorithm has observed after $n \leq m/(256 \eps)$ samples is limited and a considerable number are yet unobserved. More formally, in the first part we argue after drawing $n$ samples, with high probability $L$ has observed at most $3/4$ fraction of distinct points in $\points_G$. The proof goes as follows. Consider $O_G$ after drawing $n$ samples and let $B$ be the event that $|O_G| \geq 3|P_G|/4 = 3m/16$. Since each example is a gaming example with probability $32\eps$, the expectation of $|O_G|$ is at most $m/(256\eps) \times 32\eps = m/8$. Using Chernoff-Hoeffding bounds, $\Pr[|O_G| \geq (m/8)(1+1/2)]\leq e^{-m/96}\leq 10^{-4}$, where the last inequality holds if $m \geq 1000$.\footnote{We use the assumption of $m \geq 1000$ for the concentration bounds. For $m < 1000$, $L$ still needs to observe $\Omega(1/\eps)$ to sample a gaming example.} Therefore, $B$ happens with probability at least $0.99$.

Then in the second part, we argue if the algorithm does not include many of the unobserved points from $P_G$ in $\pfinal$ it has low performance; and if it does it has high error. Let $E$ be the event that $L$ includes at least $1/16$ fraction of the points from $P_G$ in $\pfinal$. As discussed previously $L$ does not include any points from $O_G$. Therefore, $E$ is also the event that $L$ includes at least $1/16$ fraction of $P_G \setminus O_G$ in $\pfinal$. Note that even after observing the samples, the gaming points corresponding to $P_G \setminus O_G$, are still uniformly distributed among $\points \setminus O_G$. From the setup, it is clear that the problem of predicting whether the designated target point of an example belongs to $P_G \setminus O_G$ is equivalent to predicting the outcome of a random process. To have a performance of $3/4\ \opt$, $L$ needs to include at least $3/4 \times 3/4 > 1/2$ fraction of $\points$; and since it does not include any $O_G$ it needs to include at least $1/2$ fraction of $\points \setminus O_G$, which is what we assume in the remaining part of the proof. We argue that conditioned on $B$, event $E$ happens with a high probability. Let $Z_i$ be an indicator random variable indicating whether $L$ includes $\vec{p}_i \in P_G \setminus O_G$  in $\pfinal$. Including at least $1/2$ fraction of $\points \setminus O_G$ implies $\E[Z_i] \geq 1/2$. Let $Z$ be the sum of $Z_i$ for the unobserved points in $P_G$, i.e., $Z = \sum Z_i \cdot \mathbbm{1}{[\vec{p}_i \in P_G \setminus O_G]}$. Conditioned on $B$, the event that $|O_G| \geq 3|P_G|/4$, we find $\E[Z\mid B] \geq |P_G|/8$. We have
\begin{align*} 
&\Pr[(Z<\frac{|P_G|}{16})\mid B] \leq e^{-\frac{|P_G|}{64}}\leq e^{-\frac{250}{64}} < 0.03,\\
&\Pr[E\mid B] > 0.97,\\
&\Pr[E] \geq \Pr[E\cap B] = \Pr[E\mid B]\cdot \Pr[B] \geq 0.96;
\end{align*}
where the first line follows from Chernoff-Hoeffding bounds and using $m \geq 1000$; the second line is a direct implication of the first line; and the last line uses properties of conditional probabilities and a lower bound on the probability of event $B$ that we found in the first part of the proof.

Finally, we argue that if $E$ happens, in expectation more than an $\eps$ fraction of the $m$ examples would game. Since there is a total probability mass of $32\eps$ distributed uniformly over gaming examples with designated points $\in P_G$, and $L$ includes at least $1/16$ fraction of the points from $P_G$, in expectation an $(32\eps)(1/16)=2\eps$ fraction of the examples would game.

\end{proof}

\section{Algorithmic Results Specific to the Linear Model}\label{sec:linear}

The algorithmic results provided so far work in both the general discrete and the linear discrete models. In this section we focus on the linear model and provide algorithmic results for various problems. These algorithms do not follow the greedy structure of the previous algorithms, and use novel technical ideas. First, we consider the problem of designing \emph{linear classifiers}. \Cref{sec:linear_classifier_properties} provides introductory observations and definitions about linear classifiers. \Cref{sec:linear_classifier_result} presents the main result of this section which determines whether there exists a linear classifier that classifies all agents accurately and causes all improvable agents to become qualified. \Cref{sec:linear_classifier_2dim} provides a linear classifier maximizing the number of true positives minus false positives in the two-dimensional case. Then, we shift focus to general (not necessarily linear) classifiers in a two-dimensional space and in \Cref{sec:two-dimension} provide an algorithm for maximizing true positives subject to no false positives. 

\subsection{Properties of Linear Classifiers}\label{sec:linear_classifier_properties}

Before diving into discussion of the algorithmic results, we provide observations about linear classifiers to set the context. We also provide optimal classifiers in special cases. 

For the following discussion, consider linear classifier $\f^*: \vec{a}^*\vec{x} \geq b^*$ that separates the truly qualified agents from unqualified agents.

\begin{observation}\label{obs:move} With linear classifier $\f:\vec{a}\vec{x} \geq b$, any utility maximizing agent that achieves non-negative utility by changing their features moves in dimension $\argmax_j \vecj{a}/\vecj{c}$. \end{observation}

\begin{definition}[movement dimension]\label{def:move_dim} The movement dimension of linear classifier $\f:\vec{a}\vec{x} \geq b$ is the utility maximizing dimension $\argmax_j \vecj{a}/\vecj{c}$ discussed in \Cref{obs:move}. If there are multiple such dimensions the ties are broken in favor of improvement dimensions and then lexicographically.
\end{definition}

\begin{definition}[encourage improvement/gaming]\label{def:encourage}
A classifier encourages improvement if its movement dimension is an improvement dimension. It encourages gaming otherwise.
\end{definition}

\begin{definition}[dim-$j$ improving] A linear classifier is dim-$j$ improving if it encourages improvement  and its movement dimension is along dimension $j$.
\end{definition}

The following definition captures the set of agents that potentially can improve to become truly qualified.

\begin{definition}[improvement margin, improvable agents] The improvement margin includes all the agents that can afford (do not have to incur a cost of more than $1$) to move in an improvement dimension and become truly qualified. Formally, any initially unqualified agent $i$, i.e., $\vec{a}^*\init{x}_i < b^*$, that has distance $\leq 1/\vecj{c}$ along an improvement dimension $j$ to $\f^*$ is in the improvement margin.
\end{definition}

\begin{lemma}\label{lm:all-improve}
If $\f^*:\vec{a}^* x \geq b^*$ encourages improvement, the optimal classifier is $\f^*$---among all linear or nonlinear classifiers. 
\end{lemma}

\begin{proof} 
$\f^*$ classifies initially qualified agents and unqualified unimprovable agents accurately. Also, all the agents in the improvement margin improve, become qualified, and are accurately classified as positive. 
\end{proof}

\begin{lemma}\label{lm:shifted_opt}
Let $j$ be the movement dimension of classifier $\f^*$. The classifier $\g: \vec{a}^*\vec{x} \geq b^* + \vecstarj{a}/\vecj{c}$ classifies all the initially qualified agents as positive and the rest as negative.
\end{lemma}
\begin{proof}
Initially unqualified agents, $\vec{a}^*\init{x}_i < b^*$, can move at most $1/\vecj{c}$ in dimension $j$ which is not enough to reach to $\g$. Therefore, these agents are classified as negative by $\g$. On the other hand, initially qualified agents, $\vec{a}^*\init{x}_i \geq b^*$, afford to reach to $\g$ and receive nonnegative utility. Therefore, they will be classified as positive.
\end{proof}

\begin{corollary}\label{cor:all-game}
If all the dimensions are gaming dimensions, $\g: \vec{a}^*\vec{x} \geq b^* + \vecstarj{a}/\vecj{c}$ is the optimal classifier, where $j$ is the movement dimension of $\f^*$.
\end{corollary}
\begin{proof}
If all dimensions are gaming dimensions, there are no improvable agents. Therefore, all agents are either initially qualified or unimprovable and unqualified. By \Cref{lm:shifted_opt}, $\g$ classifies all such agents accurately.
\end{proof}

By \Cref{lm:shifted_opt}, $\g: \vec{a}^*\vec{x} \geq b^* + \vecstarj{a}/\vecj{c}$ may be a ``reasonable" solution because it classifies all the initially qualified as positive and does not result in any false positive classifications. However, it misses out on any new true positives resulting from encouraging agents to become qualified. From this point on, we aim to study other classifiers (not necessarily parallel to $\f^*$) with the hope of encouraging other agents to become qualified.

\subsection{Linear Classifier for Improvable Agents}\label{sec:linear_classifier_result}

In this subsection, we study a problem that takes as input three disjoint subsets of the agents, $\mathcal{S}^\text{yes}$, $\mathcal{S}^\text{no}$, and $\mathcal{S}^\text{imp}$, and outputs a linear classifier (if one exists) that satisfies the following properties.
\begin{enumerate}[label=\roman*]
\item \label{prop:yes} Classifies agent $i$ such that $\init{x}_i \in \mathcal{S}^\text{yes}$ as positive.
\item \label{prop:no} Classifies agent $i$ such that $\init{x}_i \in \mathcal{S}^\text{no}$ as negative.
\item \label{prop:imp} Encourages agent $i$ such that $\init{x}_i \in \mathcal{S}^\text{imp}$ to improve and become truly qualified, i.e., $\true{x}_i \in \q$, and classifies $i$ as positive.
\end{enumerate}

The main result of the section is solving this problem in polynomial time. When $\mathcal{S}^\text{yes}$ is the set of initially qualified agents, $\mathcal{S}^\text{no}$ is the set of unqualified and unimprovable, and $\mathcal{S}^\text{imp}$ is the set of improvable agents, this problem determines whether there exists a linear classifier that classifies $\mathcal{S}^\text{yes}$ and $\mathcal{S}^\text{no}$ accurately and makes all the improvable agents qualified.

To solve this problem, we divide it into subproblems as following: Does there exist a linear classifier with movement direction in \emph{dimension $j$} that satisfies properties \ref{prop:yes}, \ref{prop:no}, and \ref{prop:imp}? If the answer is ``yes" for some dimension $j$, then the answer to the main problem is ``yes". If the answer is ``no" for all $1\leq j \leq d$, no linear classifier satisfying the three properties exists. 

Note that if $\mathcal{S}^\text{imp}$ is nonempty, in order to satisfy property \ref{prop:imp}, dimension $j$ must be an improvement dimension. Therefore, we study the following problem. 

\begin{problem}\label{prob:dimj}
Does there exist a dim-$j$ improving classifier (a linear classifier encouraging improvement \emph{in dimension $j$}) that satisfies properties \ref{prop:yes}, \ref{prop:no}, and \ref{prop:imp}?
\end{problem}
We propose a linear program that solves \Cref{prob:dimj}.
The following definition and observations illustrate the conditions under which a dim-$j$ improving classifier satisfies each property for agent $i$.

\begin{definition}\label{def:xif}
For a fixed improvement dimension $j$ and classifiers $\f^*:\vec{a}^*\vec{x} \geq b^* $ and $\f:\vec{a}\vec{x} \geq b $, the points $\vec{x}_{i,\f^*}$, $\vec{x}_{i,\f}$, $\vec{x}_{i,\text{max}}$ are defined as follows (depicted in \Cref{fig:improve_vs_game2}.):
\begin{itemize}
\item $\vec{x}_{i,\f^*}$ is the projection of $\init{x}_i$ on the separating hyperplane of classifier $\f^*$ along dimension $j$.
\item $\vec{x}_{i,\f}$ is the projection of $\init{x}_i$ on the separating hyperplane of classifier $\f$ along dimension $j$.
\item $\vec{x}_{i,\text{max}}$ is the shifted $\init{x}_i$ along dimension $j$ by $1/\vecj{c}$.
\end{itemize}
More formally, for all coordinates $ k\neq j$, we have $ \vec{x}_{i,\f^*}[k] = \vec{x}_{i,\f}[k] = \vec{x}_{i,\text{max}}[k] = \init{x}_i[k]$ . Also, since $\vec{a}^* \vec{x}_{i,\f^*} = b^*$, we have
$\vec{x}_{i,\f^*}[j] = {\left(b^*-\sum_{k\neq j}\vecstark{a} \init{x}_i[k]\right)}/{\vecstarj{a}}$. Similarly, since $\vec{a} \vec{x}_{i,\f} = b $, we have
$\vec{x}_{i,\f}[j] = {\left(b-\sum_{k\neq j}\vecstark{a} \init{x}_i[k]\right)}/\vecj{a}$. Finally, $\vec{x}_{i,\text{max}}[j] = \init{x}_i[j] + 1/\vecj{c}$. 

\end{definition}

\begin{observation}\label{obs:all_positive}
A dim-$j$ improving classifier $\f:\vec{a}\vec{x} \geq b$ classifies agent $i$ as positive (property \ref{prop:yes}) if $\vec{a}\vec{x}_{i,\text{max}} \geq b$. It classifies agent $i$ as negative (property \ref{prop:no}) if $\vec{a}\vec{x}_{i,\text{max}} < b$.
\end{observation}
\begin{observation}\label{obs:xif}
Using a dim-$j$ improving classifier $\f$, agent $i$ becomes qualified and is classified as positive (property \ref{prop:imp}) if and only if $\vec{x}_{i,\f^*}[j] \leq \vec{x}_{i,\f}[j] \leq \vec{x}_{i,\text{max}}[j]$. See \Cref{fig:improve_vs_game2}.
\end{observation}

\begin{figure}[ht]
\includegraphics[width=7cm]{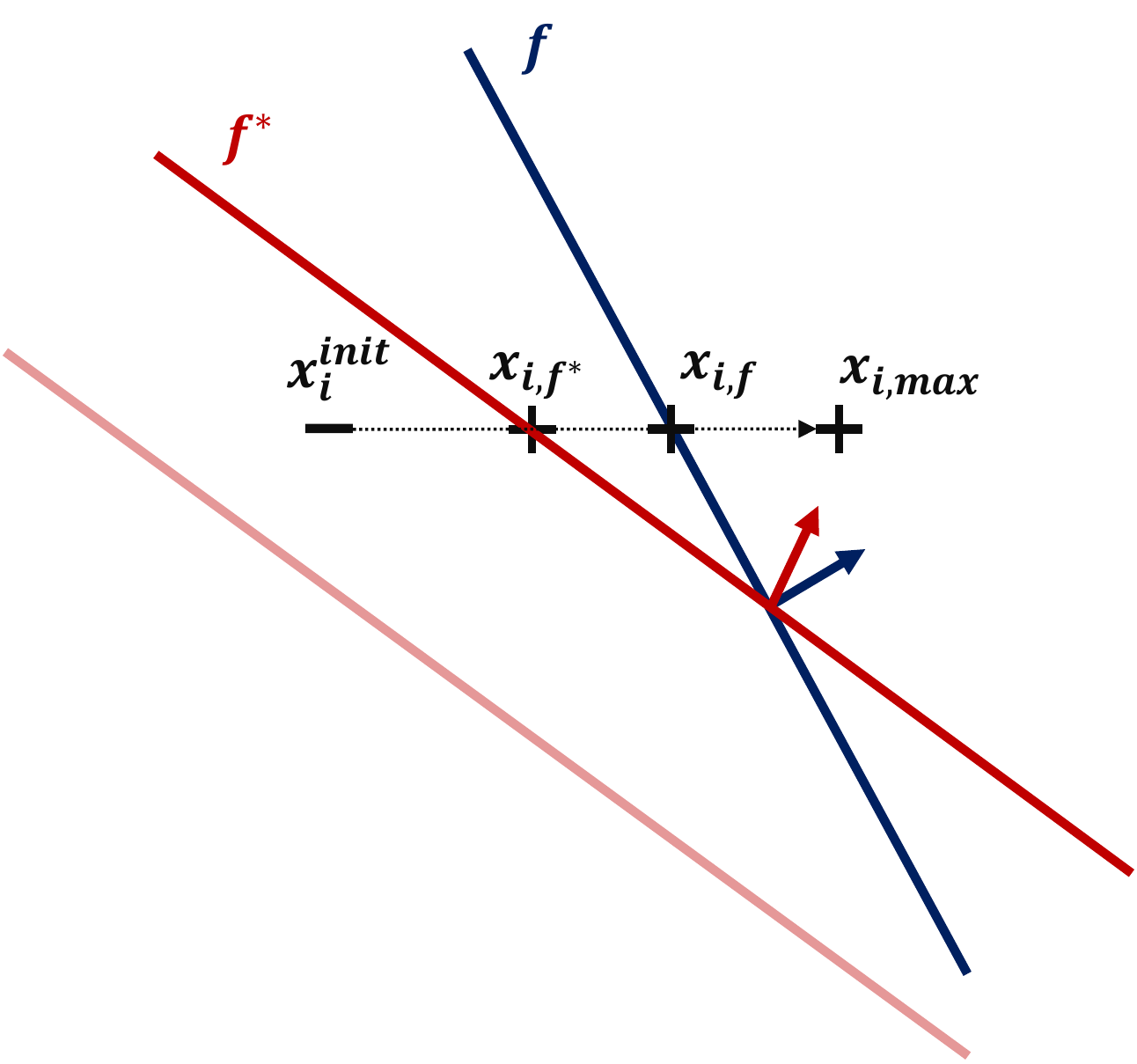}
\centering
\caption{Depicting $\init{x}_i, \vec{x}_{i,\f^*}, \vec{x}_{i,\f}, \vec{x}_{i,\text{max}}$ in \Cref{def:xif} and \Cref{obs:xif}. The horizontal axis shows dimension $j$ in the definition.} 
\label{fig:improve_vs_game2}
\end{figure}

\begin{proposition}\label{pr:three_set_LP}
The following LP captures \Cref{prob:dimj}, where the variables are $\vec{a}$ and $b$.
\begin{align}
\frac{\veck{a}}{\veck{c}} &\leq
\frac{\vecj{a}}{\vecj{c}} && \forall k\neq j \label{LP:game-vs-improve-cond1}\\
b &\leq \vec{a} \vec{x}_{i,\text{max}}   && \forall \init{x}_i \in \mathcal{S}^\text{yes}\label{LP:game-vs-improve-cond2}\\
\vec{a} \vec{x}_{i,\text{max}}  &< b && \forall \init{x}_i \in \mathcal{S}^\text{no}\label{LP:game-vs-improve-cond3}\\
\vec{x}_{i,\f^*}[j] &\leq \vec{x}_{i,\f}[j] && \forall \init{x}_i \in \mathcal{S}^\text{imp}\label{LP:game-vs-improve-cond4}\\
\vec{x}_{i,\f}[j] &\leq \vec{x}_{i,\text{max}}[j] && \forall \init{x}_i \in \mathcal{S}^\text{imp}\label{LP:game-vs-improve-cond5}
\end{align}
\end{proposition}
Constraint \ref{LP:game-vs-improve-cond1} asserts that the movement direction of the classifier is along dimension $j$. Constraint \ref{LP:game-vs-improve-cond2} asserts property \ref{prop:yes}. Constraint \ref{LP:game-vs-improve-cond3} asserts property \ref{prop:no}. Finally, constraints \ref{LP:game-vs-improve-cond4} and \ref{LP:game-vs-improve-cond5} assert property \ref{prop:imp}.

\begin{theorem}\label{thm:lp}
Given the sets $\mathcal{S}^\text{yes}$, $\mathcal{S}^\text{no}$, and $\mathcal{S}^\text{imp}$, there is a polynomial-time algorithm that outputs a linear classifier (if one exists) that satisfies Properties \ref{prop:yes}, \ref{prop:no},\ref{prop:imp}, or declares non-existence of such a classifier.
\end{theorem}
\begin{proof}
If $\mathcal{S}^\text{imp} \neq \emptyset$, run LP \ref{LP:game-vs-improve-cond1}-\ref{LP:game-vs-improve-cond5} for all improvement dimensions $j$. If $\mathcal{S}^\text{imp} = \emptyset$, run the LP for $1 \leq j \leq n$. By \Cref{pr:three_set_LP}, if there exist feasible solution $\vec{a}$ and $b$ for one of these LPs, $\f: \vec{a}\vec{x} \geq b$ is a classifier satisfying properties \ref{prop:yes}, \ref{prop:no}, and \ref{prop:imp}. 
\end{proof}
\begin{corollary}
There is a polynomial-time algorithm that determines whether there exists a linear classifier that classifies the initially qualified as positive,  unqualified unimprovable agents as negative, encourages the agents in the improvement margin to improve to become qualified, and classifies them as positive. If such a classifier exists, it maximizes true positives subject to no false positives.
\end{corollary}

\begin{remark}
\Cref{thm:relaxed-no-destination-pts} asserts that
given the initial feature vectors of agents, $\init{x}_1, \init{x}_2, \ldots, \init{x}_n\in \mathbb{R}^d$, deciding whether there exists a classifier for which all the agents become true positives is NP-hard. However, when limiting to linear classifiers this problem is no longer NP-Hard. Using \Cref{thm:lp}, by setting $\mathcal{S}^\text{yes}$ to the set of initially qualified agents, and $\mathcal{S}^\text{imp}$ to the rest of the agents, this problem is solvable in polynomial time. 
\end{remark}

\subsection{Optimal Linear Classifier in Two-Dimensional Space}\label{sec:linear_classifier_2dim}

In this subsection we continue considering linear classifiers but focus on the two-dimensional case. The main result is an algorithm (\Cref{alg:g-i-classifiers}) for finding a linear classifier that maximizes the number of true positives minus the number of false  positives.

First, note that if $\f^*: \vec{a}^*\vec{x}\geq b^*$ encourages improvement, then the optimal linear classifier is $\f^*$ (\Cref{lm:all-improve}). Also, if 
both dimensions are gaming dimensions, then a shifted $\f^*$ is optimal (\Cref{cor:all-game}). Therefore, for the rest of this subsection we focus on the case where (1) dimension one (horizontal axis) is an improvement dimension, (2) dimension two (vertical axis) is a gaming dimension, and (3) $\f^*$ encourages movement along the gaming dimension.

The following observation determines the agents that contribute to the true positives or false positives of a linear classifier.

\begin{observation}[true positives and false positives]\label{obs:true_false_positive}
Consider classifier $\g: \vec{a}\vec{x} \geq b$ that encourages improvement along dimension 1 (horizontal axis).
In \Cref{fig:true_false_positives}, we divide $\mathbb{R}^2$ into sub-areas and identify the areas of initial positions of agents that contribute to true positives and false positives. By \Cref{obs:all_positive}, each agent $i$ classified positive by $\g$ satisfies $\vec{a}\init{x}_i + \vec{a}[1]/\vec{c}[1] \geq b$. This is the set of agents with initial position at most at distance $1/\vec{c}[1]$ horizontally from $\g$ (the area on the right side of the dotted blue line in \Cref{fig:true_false_positives}).
Each agent classified as positive by $\g$, is either a true positive or a false positive. Using \Cref{obs:xif}---which identifies the initial positions of agents that improve, become qualified, and get classified as positive---any part above (with higher value in dimension $2$ than) the intersection of $\g$ and $\f^*$, on the right of the dotted blue line contributes to true positives (depicted by ``$+$" in the figure). On the other hand, any part below (with less value in dimension $2$ than) the intersection of $\g$ and $\f^*$, on right of the dotted blue line, and the left side of $\f^*$ contributes to false positives (depicted by ``$-$" in the figure).
\end{observation}

\begin{figure}[ht]
\includegraphics[width=5cm]{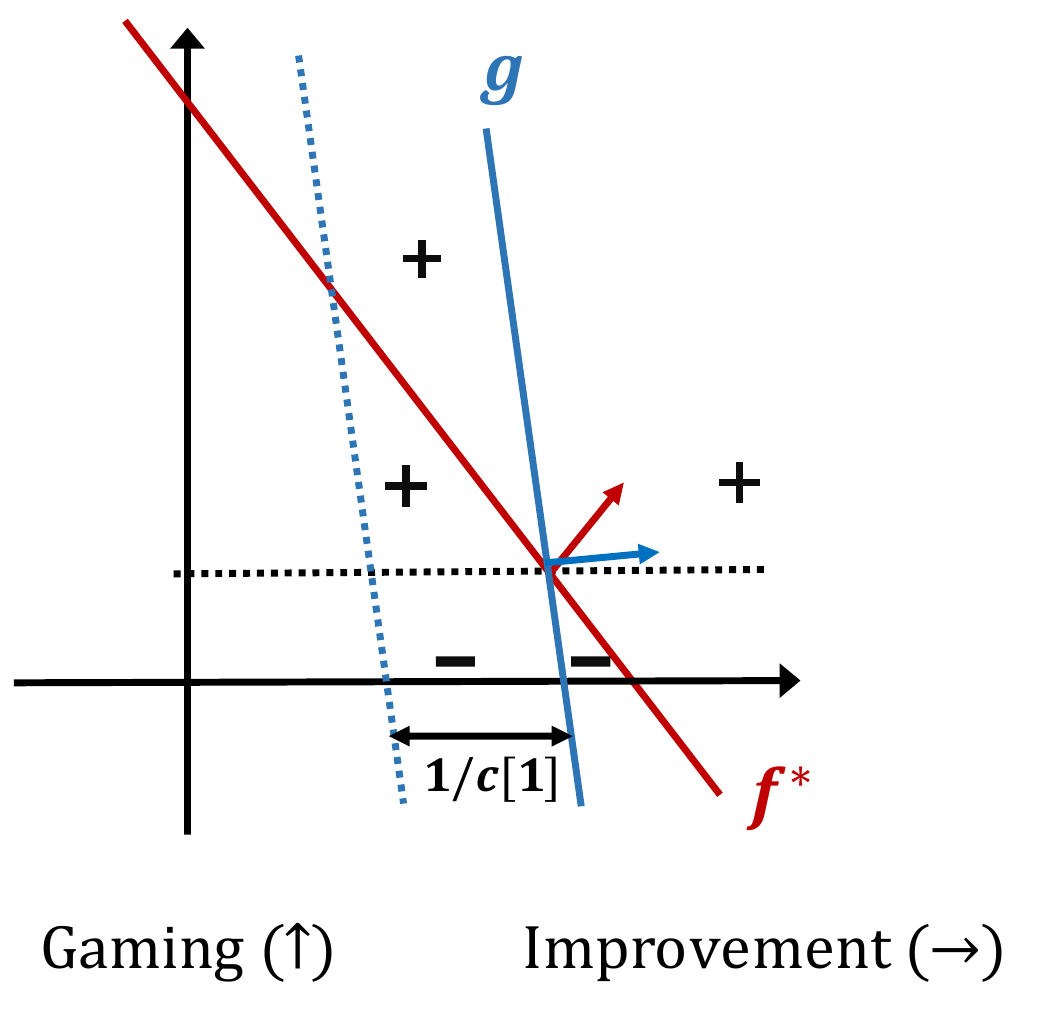}
\centering
\caption{Related to \Cref{obs:true_false_positive}. This figure identifies the areas of initial position of agents that contribute to true positives, depicted by ``$+$", and false positives, depicted by ``$-$", when using classifier $\g$ that encourages improvement along dimension $1$ (horizontal axis).}
\label{fig:true_false_positives}
\end{figure}

In the following lemma, we find a dominating set of linear classifiers.

\begin{lemma}\label{lm:best_slope}
There exists a linear classifier  $\g:\vec{a}\vec{x} \geq b$ satisfying $\vecone{a}/\vecone{c} = \vectwo{a}/\vectwo{c}$ that  maximizes the number of true positives minus false positives in $\mathbb{R}^2$ among improvement-encouraging linear classifiers.
\end{lemma}
\begin{proof}
Consider an arbitrary linear classifier $\g':\vec{a}'\vec{x} \geq b'$ that encourages improvement. By \Cref{obs:move}, $\vec{a}'[1]/\vec{c}[1]\geq \vec{a}'[2]/\vec{c}[2]$. Consider point $\vec{z}$ where $\g'$ intersects with $\f^*$. Let $\g:\vec{a}\vec{x} \geq b$ be the classifier satisfying $\vecone{a}/\vecone{c} = \vectwo{a}/\vectwo{c}$ that passes through $\vec{z}$. We show $\g$ has objective value at least as that of $\g'$. 
Note that, by construction, $\g$ lies between $\g'$ and $\f^*$ in $\mathbb{R}^2$; in \Cref{fig:optimal_slope}, the blue classifier, the green classifier, and the red classifier illustrate $\g'$, $g$, and $\f^*$, respectively. \Cref{obs:true_false_positive} determines the areas of true positives and false positives based on the initial positions of the agents and the classifier in use. By \Cref{obs:true_false_positive}, $\g$ includes all the true positives of $\g'$ as well as the area between the dotted blue and green lines above their intersection; this part is illustrated by ``$+$" in \Cref{fig:optimal_slope}. Also, $\g'$ includes all the false positives of $\g$ as well as the area between the blue and green dotted lines below their intersection; this part is illustrated by ``$-$" in \Cref{fig:optimal_slope}. Since the true positive area of $\g$ is a superset and its false positive area is a subset compared to $\g'$, it has weakly higher objective value.
\end{proof}

\begin{figure}[ht]
\includegraphics[width=5cm]{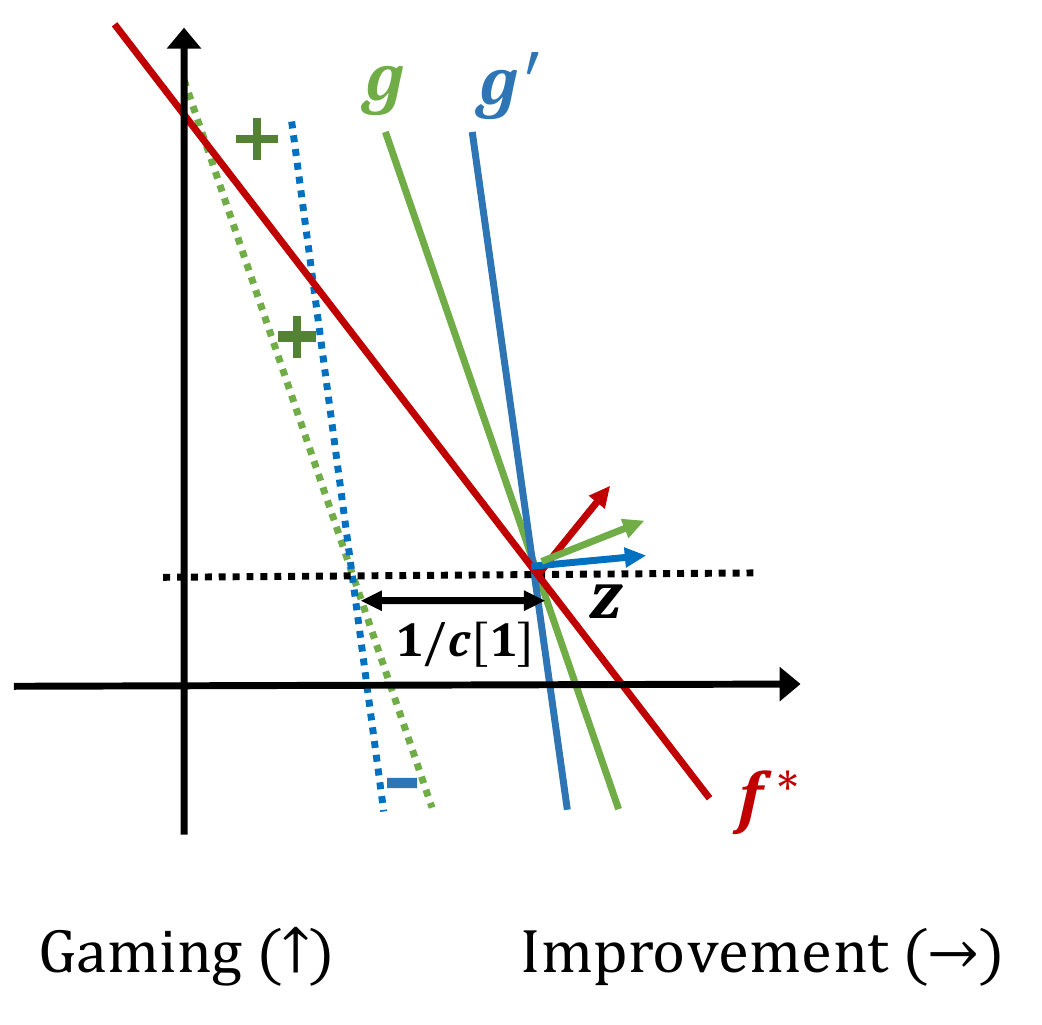}
\centering
\caption{Related to \Cref{lm:best_slope}. $\g$ and $\g'$ are classifiers that encourage improvement along dimension $1$ (horizontal axis). The area between the dotted blue and green lines above the intersection shows the area of initial positions of agents that contribute to true positives by $\g$ and not by $\g'$---it is denoted by ``$+$". The area between the two lines below the intersection shows the area of initial positions of agents that contribute to false positives by $\g'$ and not by $\g$---it is denoted by ``$-$".}
\label{fig:optimal_slope}
\end{figure}

\textbf{Overview of \Cref{alg:g-i-classifiers}.}  The algorithm finds the best of gaming-encouraging and improvement-encouraging linear classifiers, and outputs the better of the two in terms of the objective function. The best gaming-encouraging classifier is given in \Cref{cor:all-game}. The best improvement-encouraging classifier is found as follows: Using \Cref{lm:best_slope}, the optimal slope of the classifier is known. Therefore, we only need to determine a crossing point to determine the classifier. For all agents $i$, consider linear classifier $g:\vec{c}^*\vec{x} \geq \vec{c}\init{x}_i + 1$. This classifier satisfies the optimal slope from \Cref{lm:best_slope} and is at distance $1/\vecone{c}$ horizontally from $\init{x}_i$. Therefore, $i$ is the farthest agent to reach the classifier and be classified positive. Among these classifier, find one that maximizes the objective value---this is the optimum improvement classifier.

\begin{algorithm}[!ht]
    \SetNoFillComment
    \SetAlgoLined
    \SetKwInOut{Input}{Input}
    \SetKwInOut{Output}{Output}
    \SetKw{Return}{return}
     \DontPrintSemicolon
    \Input{Initial positions of agents:  $\mathcal{X} \subseteq \mathbb{R}^2$ \tcp{dim1 is improvement. dim2 is gaming.}}
    \Input{Linear model: $\f^*:\vec{a}^* \vec{x} \geq b^*$ \tcp{$\f^*$ encourages movement in dim2.}}
            \Output{$\f$ \tcp{best linear classifier}}
    $\text{objective} = 0$ \tcp{maximum objective value observed}
    $f:\vec{a}^*\vec{x} \geq b^* +  \vec{a}^*[2]/\vectwo{c}$\ \tcp{best gaming-encouraging linear classifier}
    $\text{objective} = \text{count\_true\_positives}(f) - \text{count\_false\_positives}(f)$\;

    \For{$i=1, 2, \cdots, |\mathcal{X}|$}{ 
        $g:\vec{c}\vec{x} \geq \vec{c}\init{x}_i + 1$ \tcp{ classifier with optimal slope corresponding to agent $i$}
        $\text{diff} = \text{count\_true\_positives}(g) - \text{count\_false\_positives}(g)$\;
        \If{$\text{diff} > \text{objective}$}{
            $\text{objective} = \text{diff}$\;
            $f = g$\;
        }
    }
    \Return $f$\;
    \caption{Find a linear classifier maximizing \#true-positives minus \#false-positives}
    \label{alg:g-i-classifiers}
\end{algorithm}

\begin{theorem}
\Cref{alg:g-i-classifiers} finds the linear classifier maximizing the number of true positives minus false positives in $\mathbb{R}^2$, when dimension 1 is improvement and dimension 2 is gaming. When both dimensions are improvement or gaming, the optimal classifiers are given in \Cref{lm:all-improve,cor:all-game}.
\end{theorem}
\begin{proof}
The algorithms considers the best of two groups of linear classifiers---gaming-encouraging and improvement-encouraging. The best gaming-encouraging classifier is given in \Cref{cor:all-game}. Among the improvement-encouraging classifiers, \Cref{lm:best_slope} gives the optimal slope of the classifier. We claim we only need to consider classifiers with the optimal slope such that some agent $i$ is exactly at distance $1/\vecone{c}$ horizontally from the classifier -- a classifier where $i$ is the farthest agent that will be classified positive. Consider the set of such parallel classifiers $\f_1, \f_2, \ldots, \f_m$, where they are sorted based on their $y$-intercept such that $\f_1$ is the classifier corresponding to agent $i = \argmax \vec{c}\init{x}_i$, and $\f_m$ is the classifier corresponding to agent $i = \argmin \vec{c}\init{x}_i$. To prove the claim, we show any other classifier with the optimal slope is dominated by the classifiers we consider. First, note that since any potential parallel classifier strictly between $\f_j$ and $\f_{j+1}$ classifies the agents exactly as in $\f_j$, these classifiers are dominated and we do not need to consider them. Secondly, we do not need to consider any other parallel classifier with $y$-intercept less than $\f_m$ or higher than $\f_1$. Any parallel classifier with smaller intercept than $\f_m$ has the same classification as $\f_m$ (classifying all agents as positive). Also, any parallel classifier with larger intercept than $\f_1$ does not classify any agents as positive; therefore, has objective value $0$, and is dominated in terms of objective value by the best gaming-encouraging classifier.
\end{proof}

\begin{remark}
\Cref{thm:atmost_k_game} implies maximizing the objective function considered in this subsection --maximizing the number of true positives minus false positives-- is NP-hard in $\mathbb{R}^d$ when $d$ is not a constant and we are not limited to linear classifiers.
\end{remark}
\subsection{Optimal General Classifier in Two-Dimensional Space}
\label{sec:two-dimension}

In this subsection, we consider the problem of maximizing true positives subject to no false positives in a $2$-dimensional space, where the horizontal dimension is improvement, and the vertical dimension is gaming.
We provide an algorithm in the linear model that given a set of agents, returns a set of target points $\pfinal\subset \mathbb{R}^2$ that maximizes true positives subject to no false positives. 
Note that unlike \Cref{alg:maximal}, our algorithm in this subsection does not take a finite set of target points $\points$ as input.
For simplicity, by scaling we may assume wlog that $c = \vec{c}[1]= \vec{c}[2]$.

\textbf{Overview of \Cref{alg:2-dimension}.} 
First, all the points $\init{x}_{i}$ for $1\leq i \leq m$ are sorted along the gaming dimension in a descending order, such that $\init{x}_{n}$ has the smallest value in the gaming dimension. Our goal is to find designated points, $\vec{x'}_i$, for each $\init{x}_i$. Starting with $\init{x}_n$, for each point $\init{x}_i$, move $\init{x}_i$ along the improvement dimension until it crosses the line $\vec{a^*}\vec{x}=b^*$ at $\vec{x}_{i,min}$ (See \cref{fig:2-dim-alg}). Let $\vec{x'}_i$, the designated point of $\init{x}_i$, be initially $\vec{x'}_i = \vec{x}_{i,min}$. If given the current set of designated points for agents $n, n-1, \ldots, i$, another point $\init{x}_j$ for $j>i$ maximizes utility by moving to $\vec{x'}_i$ and becomes false positive, push $\vec{x'}_i$ upward along the gaming dimension, until $\init{x}_j$ no longer picks $\vec{x'}_i$. When pushing $\vec{x'}_i$ along the gaming dimension, let $\vec{x}_{i,max}$ denote the furthest point that $\init{x}_i$ can afford to reach to it. If the final point $\vec{x'}_i$ is such that $\init{x}_i$ cannot afford to move to it, i.e. $\vec{x'}_i[2]>\vec{x}_{i,max}[2]$, discard $\vec{x'}_i$. Otherwise, $\vec{x'}_i$ is added to $\pfinal$.

Note that we assume that if a point $\init{x}_j$ can improve to $\vec{x'}_j$ and game to $\vec{x'}_i$ with the same cost, it would pick the improvement option.

\begin{figure}[ht!]
    \centering
    \includegraphics[width=9cm]{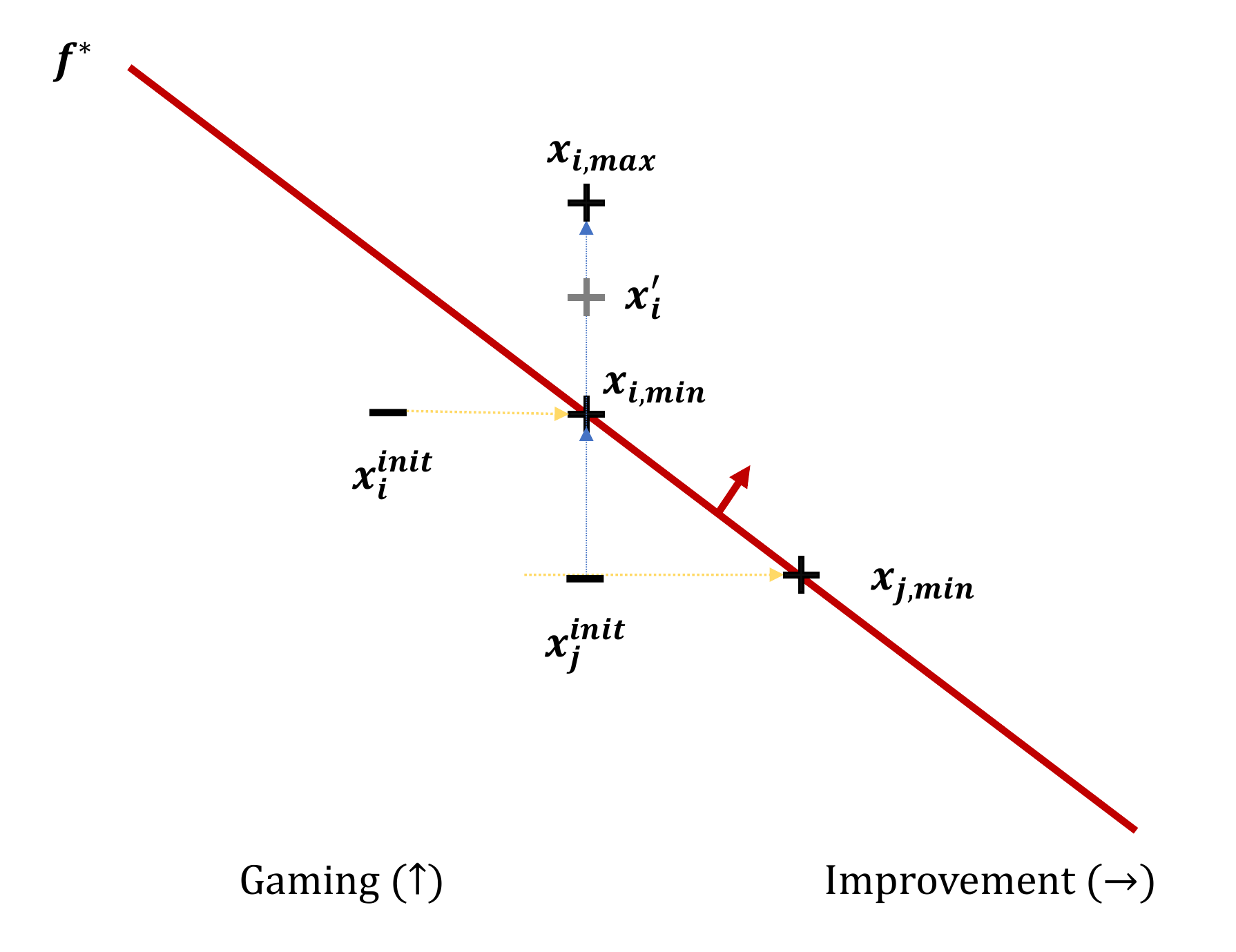}
    \caption{In \Cref{alg:2-dimension}, $\vec{x'}_i$ is pushed along the gaming dimension so $\init{x}_j$ no longer moves to it.}
    \label{fig:2-dim-alg}
\end{figure}

\begin{algorithm}
    \SetNoFillComment
    \SetAlgoLined
    \SetKwInOut{Input}{Input}
    \SetKwInOut{Output}{Output}
    \SetKw{Continue}{continue}
    \SetKw{Return}{return}
    \Input{$\agents$, $f^*: \vec{a}^*\vec{x}\geq b^*$}
    \Output{$\mathcal{P}^{\text{final}}$}
    Sort $\vec{x}_i\in \agents$ in a descending order of $\vec{x}_i[2]$\;
    \For{$i=n,\cdots,1$}{
        \tcc{Let $\vec{x}_{i,min}$ be the projection of $\vec{x}_i$ on $\vec{a}^*\vec{x}=b^*$ along the improvement dimension}
        $\vec{x}_{i,min} = \Big(\frac{b^*-\vec{a}^*[2]\vec{x}_i[2]}{\vec{a^*}[1]},\vec{x}_i[2]\Big)$\;
        \If{$\vec{x}_{i,min}[1]-\vec{x}_i[1]>1/c$}{
        \tcc{$\vec{x}_i$ cannot become true positive.}
            \Continue\;
        }
        $\vec{x'}_i \leftarrow \vec{x}_{i,min}$\;
        \For{$j=n,\cdots,i+1$}{
            \If{$cost(\vec{x}_j,\vec{x}'_j)>cost(\vec{x}_j,\vec{x}'_i)$}{
                $\vec{x}'_i \leftarrow (\vec{x}'_i[1], \vec{x}'_i[2]+cost(\vec{x}_j,\vec{x}'_j)-cost(\vec{x}_j,\vec{x}'_i))$\;
            }
        }
        \If{$\vec{x'}_i[2] > \vec{x}_{i,max}[2]$}{
            \tcc{$\vec{x}_i$ cannot become true positive without another point becoming false positive.}
           $\vec{x'}_i = (\vec{x'}_i[1], \infty)$\;
        }
        $\mathcal{P}^{\text{final}}\leftarrow \mathcal{P}^{\text{final}}\cup \vec{x'}_{i}$\;
        \Return $\mathcal{P}^{\text{final}}$\;
    }

\caption{Maximizing the number of true positives in $2$-dimensions.}   
\label{alg:2-dimension}
\end{algorithm}
In order to show that \Cref{alg:2-dimension} maximizes true positives subject to no false positives, we need the following observation and lemma.

\begin{observation}
Line $\vec{a}^*\vec{x}=b^*$ has a negative slope, i.e., each feature is defined so that larger is better. Therefore, after the points in $\agents$ are sorted, if an agent $\init{x}_j$ where $j<i$ reaches to any point $\vec{x'}_i\in [\vec{x}_{i,min}, \vec{x}_{i,max}]$, then $\init{x}_j$ becomes true positive. On the other hand, for $j>i$, if $\init{x}_j$ moves to any point $\vec{x'}_i\in[\vec{x}_{i,min}, \vec{x}_{i,max}]$, then $\vec{x}_j$ becomes false positive.
\label{observation:2-dimension-negative-slope}
\end{observation}

\begin{lemma}
Consider a point $\vec{p}$ such that $\vec{p}[1]\geq \vec{x}_{i,min}[1]$, and another point $\vec{q}\in [\vec{x}_{i,min},\vec{x}_{i,max}]$. Suppose $cost(\init{x}_i, \vec{p}) = cost(\init{x}_i, \vec{q})$. Then, for any $j>i$, it is the case that $cost(\init{x}_j, \vec{p})\leq cost(\init{x}_j, \vec{q})$.
\label{lemma:2-dim-triangle-inequalities}
\end{lemma}

\begin{proof}
Proof is deferred to \Cref{appendix:proof-lemma-2-dim-triangle-inequalities}.
\end{proof}

\begin{theorem}
Given initial feature vectors of agents, $\init{x}_1, \init{x}_2, \ldots, \init{x}_n \in \mathbb{R}^2$, \Cref{alg:2-dimension} maximizes the number of true positives subject to no false positives.
\end{theorem}

\begin{proof}
Suppose not.  Let $\vec{x}_1^{\opt}, \ldots,\vec{x}_n^{\opt}$ be an optimal solution that agrees with $\vec{x'}_1, \ldots, \vec{x'}_n$ on as large a suffix as possible, and let $i$ be the largest index such that $\vec{x}_i^{\opt}\neq \vec{x}_i'$ (so $\vec{x}_j^{\opt}= \vec{x}_j'$ for all $j>i$).  

First, note that $i\neq n$.  This is because $\vec{x'}_n = \vec{x}_{n,min}$, which is the cheapest point that agent $n$ can reach to become a true positive; moreover, any other point moving to $\vec{x}_n'$ is a true improvement.  So, replacing $\vec{x}_n^{\opt}$ with $\vec{x'}_n$ only helps.

Next, we claim that even if $i<n$, replacing $\vec{x}_i^{\opt}$ with $\vec{x'}_i$ can only improve the optimal solution.  First, if $cost(\init{x}_i, \vec{x}_i^{\opt}) \geq cost(\init{x}_i, \vec{x'}_i)$ then replacing $\vec{x}_i^{\opt}$ with $\vec{x'}_i$ only helps by the same argument as above and the fact that $\vec{x'}_i$ was chosen so that no agent $j>i$ manipulates to it; here we are using the fact that the suffixes of the two solutions agree.  On the other hand, suppose that $cost(\init{x}_i, \vec{x}_i^{\opt}) <  cost(\init{x}_i, \vec{x'}_i)$ and $cost(\init{x}_i, \vec{x}_i^{\opt}) \leq 1/c$.  Since $\init{x}_i$ cannot become a false positive by moving to $\vec{x}_i^{\opt}$, this means that  $\vec{x}_i^{\opt}[1]\geq \vec{x}_{i,min}[1]$. There exists a point $\vec{q}\in [\vec{x}_{i,min}, \vec{x}_{i,max}]$ such that $cost(\init{x}_{i}, \vec{x}_i^{\opt}) = cost(\init{x}_{i}, \vec{q})$, which implies that $cost(\init{x}_{i}, \vec{q}) < cost(\init{x}_{i}, \vec{x'}_{i})$. The reason that $\vec{q}$ was not selected as $\vec{x'}_i$ is that there exists an agent $\init{x}_j$ where $\init{x}_j$ moves to $\vec{q}$ and becomes false positive. By \Cref{observation:2-dimension-negative-slope},  $j>i$. Hence, $cost(\init{x}_j, \vec{q})<cost(\init{x}_j, \vec{x'}_j)$ and $cost(\init{x}_j, \vec{q}) \leq 1/c$.
By \Cref{lemma:2-dim-triangle-inequalities}, $cost(\init{x}_j, \vec{x}_i^{\opt})\leq cost(\init{x}_j, \vec{q})$, so
$cost(\init{x}_j, \vec{x}_i^{\opt}) < cost(\init{x}_j, \vec{x'}_j)$ and $cost(\init{x}_j, \vec{x}_i^{\opt}) \leq 1/c$.
Hence, $\init{x}_j$  is closer to $\vec{x}_i^{\opt}$ compared to $\vec{x'}_j=\vec{x}_j^{\opt}$ and so agent $j$ would become a false positive under $\opt$, which contradicts the definition of $\opt$. 
So, this second case cannot occur.  

Therefore, \Cref{alg:2-dimension} maximizes the number of true positives subject to having no false positives.
\end{proof}

\begin{remark}
By \Cref{cor:d-dim-max-TP-no-FP}, this problem is NP-hard when $\agents \subset \mathbb{R}^d$ for general (not constant) $d$.
\end{remark}

 \bibliographystyle{plainnat}
\newpage
\footnotesize{
\bibliography{ref}

\begin{thebibliography}{17}
\providecommand{\natexlab}[1]{#1}
\providecommand{\url}[1]{\texttt{#1}}
\expandafter\ifx\csname urlstyle\endcsname\relax
  \providecommand{\doi}[1]{doi: #1}\else
  \providecommand{\doi}{doi: \begingroup \urlstyle{rm}\Url}\fi

\bibitem[Ahmadi et~al.(2021)Ahmadi, Beyhaghi, Blum, and
  Naggita]{Ahmadi2021TheSP}
Saba Ahmadi, Hedyeh Beyhaghi, Avrim Blum, and Keziah Naggita.
\newblock The strategic perceptron.
\newblock In \emph{Proceedings of the 22nd ACM Conference on Economics and
  Computation}, page 6–25, New York, NY, USA, 2021. Association for Computing
  Machinery.
\newblock ISBN 9781450385541.
\newblock URL \url{https://doi.org/10.1145/3465456.3467629}.

\bibitem[Alon et~al.(2020)Alon, Dobson, Procaccia, Talgam-Cohen, and
  Tucker-Foltz]{Alon2020MultiagentEM}
Tal Alon, Magdalen Dobson, Ariel Procaccia, Inbal Talgam-Cohen, and Jamie
  Tucker-Foltz.
\newblock Multiagent evaluation mechanisms.
\newblock \emph{In Proceedings of the AAAI Conference on Artificial
  Intelligence}, 34\penalty0 (02):\penalty0 1774--1781, Apr. 2020.
\newblock \doi{10.1609/aaai.v34i02.5543}.
\newblock URL \url{https://ojs.aaai.org/index.php/AAAI/article/view/5543}.

\bibitem[Bechavod et~al.(2020)Bechavod, Ligett, Wu, and
  Ziani]{Bechavod2020CausalFD}
Yahav Bechavod, Katrina Ligett, Zhiwei~Steven Wu, and Juba Ziani.
\newblock Causal feature discovery through strategic modification.
\newblock \emph{ArXiv}, abs/2002.07024, 2020.
\newblock URL \url{https://arxiv.org/abs/2002.07024}.

\bibitem[Braverman and Garg(2020)]{Braverman2020TheRO}
Mark Braverman and Sumegha Garg.
\newblock The role of randomness and noise in strategic classification.
\newblock In \emph{Proceedings of the 1st Symposium on Foundations of
  Responsible Computing, {FORC} 2020, June 1-3, 2020, Harvard University,
  Cambridge, MA, {USA} (virtual conference)}, volume 156 of \emph{LIPIcs},
  pages 9:1--9:20. Schloss Dagstuhl - Leibniz-Zentrum f{\"{u}}r Informatik,
  2020.
\newblock \doi{10.4230/LIPIcs.FORC.2020.9}.
\newblock URL \url{https://doi.org/10.4230/LIPIcs.FORC.2020.9}.

\bibitem[Br\"{u}ckner and Scheffer(2011)]{adversarial_games_pred}
Michael Br\"{u}ckner and Tobias Scheffer.
\newblock Stackelberg games for adversarial prediction problems.
\newblock In \emph{Proceedings of the 17th ACM SIGKDD International Conference
  on Knowledge Discovery and Data Mining}, KDD ’11, page 547–555, New York,
  NY, USA, 2011. Association for Computing Machinery.
\newblock ISBN 9781450308137.
\newblock \doi{10.1145/2020408.2020495}.
\newblock URL \url{https://doi.org/10.1145/2020408.2020495}.

\bibitem[Dong et~al.(2018)Dong, Roth, Schutzman, Waggoner, and
  Wu]{revealed_preferences}
Jinshuo Dong, Aaron Roth, Zachary Schutzman, Bo~Waggoner, and Zhiwei~Steven Wu.
\newblock Strategic classification from revealed preferences.
\newblock In \emph{Proceedings of the 2018 ACM Conference on Economics and
  Computation}, EC ’18, page 55–70, New York, NY, USA, 2018. Association
  for Computing Machinery.
\newblock ISBN 9781450358293.
\newblock \doi{10.1145/3219166.3219193}.
\newblock URL \url{https://doi.org/10.1145/3219166.3219193}.

\bibitem[Feige(1998)]{feige1998threshold}
Uriel Feige.
\newblock A threshold of ln n for approximating set cover.
\newblock \emph{J. ACM}, 45\penalty0 (4):\penalty0 634–652, 1998.
\newblock ISSN 0004-5411.
\newblock \doi{10.1145/285055.285059}.
\newblock URL \url{https://doi.org/10.1145/285055.285059}.

\bibitem[Frankel and Kartik(2019)]{Frankel2019ImprovingIF}
Alex~M. Frankel and Navin Kartik.
\newblock Improving information from manipulable data.
\newblock \emph{arXiv: Theoretical Economics}, 06 2019.
\newblock ISSN 1542-4766.
\newblock \doi{10.1093/jeea/jvab017}.
\newblock URL \url{https://doi.org/10.1093/jeea/jvab017}.

\bibitem[Haghtalab et~al.(2020)Haghtalab, Immorlica, Lucier, and
  Wang]{Haghtalab2020MaximizingWW}
Nika Haghtalab, Nicole Immorlica, Brendan Lucier, and Jack~Z. Wang.
\newblock Maximizing welfare with incentive-aware evaluation mechanisms.
\newblock In Christian Bessiere, editor, \emph{Proceedings of the Twenty-Ninth
  International Joint Conference on Artificial Intelligence, {IJCAI-20}}, pages
  160--166. International Joint Conferences on Artificial Intelligence
  Organization, 7 2020.
\newblock \doi{10.24963/ijcai.2020/23}.
\newblock URL \url{https://doi.org/10.24963/ijcai.2020/23}.
\newblock Main track.

\bibitem[Hardt et~al.(2016)Hardt, Megiddo, Papadimitriou, and
  Wootters]{Hardt2016}
Moritz Hardt, Nimrod Megiddo, Christos Papadimitriou, and Mary Wootters.
\newblock Strategic classification.
\newblock In \emph{Proceedings of the 2016 ACM Conference on Innovations in
  Theoretical Computer Science}, ITCS ’16, page 111–122, New York, NY, USA,
  2016. Association for Computing Machinery.
\newblock ISBN 9781450340571.
\newblock \doi{10.1145/2840728.2840730}.
\newblock URL \url{https://doi.org/10.1145/2840728.2840730}.

\bibitem[Harris et~al.(2021)Harris, Heidari, and Wu]{harris2021stateful}
Keegan Harris, Hoda Heidari, and Zhiwei~Steven Wu.
\newblock Stateful strategic regression.
\newblock \emph{CoRR}, abs/2106.03827, 2021.
\newblock URL \url{https://arxiv.org/abs/2106.03827}.

\bibitem[Hu et~al.(2019)Hu, Immorlica, and Vaughan]{Hu:2019:}
Lily Hu, Nicole Immorlica, and Jennifer~Wortman Vaughan.
\newblock The disparate effects of strategic manipulation.
\newblock In \emph{Proceedings of the Conference on Fairness, Accountability,
  and Transparency}, FAT* '19, pages 259--268, New York, NY, USA, 2019. ACM.
\newblock ISBN 978-1-4503-6125-5.
\newblock \doi{10.1145/3287560.3287597}.
\newblock URL \url{http://doi.acm.org/10.1145/3287560.3287597}.

\bibitem[Kleinberg and Raghavan(2019)]{Kleinberg2018HowDC}
Jon Kleinberg and Manish Raghavan.
\newblock How do classifiers induce agents to invest effort strategically?
\newblock In \emph{Proceedings of the 2019 ACM Conference on Economics and
  Computation}, EC '19, page 825–844, New York, NY, USA, 2019. Association
  for Computing Machinery.
\newblock ISBN 9781450367929.
\newblock \doi{10.1145/3328526.3329584}.
\newblock URL \url{https://doi.org/10.1145/3328526.3329584}.

\bibitem[Miller et~al.(2020)Miller, Milli, and Hardt]{Miller2019StrategicCI}
John Miller, Smitha Milli, and Moritz Hardt.
\newblock Strategic classification is causal modeling in disguise.
\newblock In \emph{Proceedings of the 37th International Conference on Machine
  Learning, ICML 2020, 13-18 July 2020, Virtual Event}, volume 119 of
  \emph{Proceedings of Machine Learning Research}, pages 6917--6926. PMLR,
  2020.
\newblock URL \url{http://proceedings.mlr.press/v119/miller20b.html}.

\bibitem[Milli et~al.(2019)Milli, Miller, Dragan, and Hardt]{Milli2018TheSC}
Smitha Milli, John Miller, Anca~D. Dragan, and Moritz Hardt.
\newblock The social cost of strategic classification.
\newblock In \emph{Proceedings of the Conference on Fairness, Accountability,
  and Transparency}, FAT* '19, page 230–239, New York, NY, USA, 2019.
  Association for Computing Machinery.
\newblock ISBN 9781450361255.
\newblock \doi{10.1145/3287560.3287576}.
\newblock URL \url{https://doi.org/10.1145/3287560.3287576}.

\bibitem[Shavit et~al.(2020)Shavit, Edelman, and Axelrod]{Shavit2020LearningFS}
Yonadav Shavit, Benjamin Edelman, and Brian Axelrod.
\newblock Learning from strategic agents: Accuracy, improvement, and causality.
\newblock In Hal~Daumé III and Aarti Singh, editors, \emph{Proceedings of the
  37th International Conference on Machine Learning}, volume abs/2002.10066 of
  \emph{Proceedings of Machine Learning Research}, pages 8676--8686. PMLR,
  13--18 Jul 2020.
\newblock URL \url{http://proceedings.mlr.press/v119/shavit20a.html}.

\bibitem[Xiao et~al.(2020)Xiao, Wang, Chen, Tang, and Yang]{xiao2020optimal}
Shenke Xiao, Zihe Wang, Mengjing Chen, Pingzhong Tang, and Xiwang Yang.
\newblock Optimal common contract with heterogeneous agents.
\newblock \emph{Proceedings of the AAAI Conference on Artificial Intelligence},
  34\penalty0 (05):\penalty0 7309--7316, Apr. 2020.
\newblock \doi{10.1609/aaai.v34i05.6224}.
\newblock URL \url{https://ojs.aaai.org/index.php/AAAI/article/view/6224}.

\end{thebibliography}
}

\newpage
\normalsize
\appendix
\section{Missing Proofs of \Cref{sec:general}}\label{app:hardness}
\begin{proof}[Proof of~\Cref{prop:running-time-greedy-alg}]
The size of $\agents$ is $n$, and within the for loop each computation takes $O(1)$ time since the edges for each  $x_i$ are already sorted. When the flag is set to $1$, at least one point in $\points$ is removed, and when the flag is $0$ at the end of the inner loop, the algorithm returns. Therefore, the outer loop is run at most $|\points|$ times while the inner loop is run $n$ times; resulting in a running time of $O(|\points|n)$.
\end{proof}

\begin{proof}[Proof of \Cref{thm:atmost_k_game}]
We show the following problem is NP-hard.

\begin{problem}
\label{pr:atmost_k_game}
Suppose we are given a set of $n$ agents where $\init{x}_1, \init{x}_2, \ldots, \init{x}_n$ denote their initial feature vectors, and a set $\points$ of potential criteria also called target points in the linear model. Find a subset $\pfinal \subseteq \points$ that maximizes the number of true positives subject to at most $k$ false positives. 
\end{problem}

We prove the NP-hardness by reducing the Max-$k$-Cover problem with equal-sized sets of size $3$ to this problem. In the Max-$k$-Cover problem, we are given a set $\mathcal{E}$ of elements $e_i$, and sets $S_j \subseteq \mathcal{E}$, and the goal is to select at most $k$ sets out of $S_j$ that maximize the number of elements they cover. 

First, we show how to construct an instance of \Cref{pr:atmost_k_game} from an instance of the Max-$k$-Cover problem. To do so, we determine the number of dimensions, initial positions of the agents, the target points, and the movement costs. Let $n$ be the number of elements of the Max-$k$-Cover instance, we construct an $n+1$-dimensional space where the first $n$ dimensions are improvement and the last dimension is gaming. Consider elements $e_1, e_2, \ldots, e_n$ in the Max-$k$-Cover instance. For every element, we consider an agent; and for every set, we consider an agent and a target point. For $e_i$, the corresponding agent is at initial point $\init{x}_i$, an $n+1$-dimensional vector whose $i^{th}$ and $n+1^{st}$ coordinates are $1$ and the other coordinates are $0$. For every set $S_j$, we consider a target point $\vec{p}_j$ and an agent with initial point $\init{x}_{n+j}$. In $\vec{p}_j$, the coordinates corresponding to the elements in $S_j$ and the $n+1^{st}$ coordinate are set to $1$ and the rest of the coordinates are $0$. In $\init{x}_{n+j}$, the coordinates corresponding to the elements in $S_j$ are set to $1$, the $n+1^{st}$ coordinate is set to $-1$, and the rest of the coordinates are $0$. Finally, let the movement cost in any dimension be  $1/2$. Note that this construction fits into the framework of a linear model and $\f^*:\sum_{j=1}^{n+1} \vec{x}[j] \geq 4$ is the linear threshold function for the truly qualified agents. All the target points $\vec{p}_j$ satisfy the threshold and all the agents are initially unqualified and do not meet the threshold. 

Next, we discuss what target point each agent selects and whether they become truly qualified (true positive) or not (false positive). Because the cost per unit of movement equals $1/2$, each agent can only afford to reach to target points with distance at most $2$. Agents $\init{x}_i$ for $i \in \{1,\ldots, n\}$ can only afford to reach a target point whose $i^{th}$ coordinate is $1$ since they are at distance $2$. They are at distance $3$ to any other target points. Since all dimensions $1,\ldots,n$ are improving dimensions these agents become truly qualified when they reach such target points. Agents $\init{x}_i$ for $i>n$ can only afford to reach $\vec{p}_i$ since they have distance $2$. They have distance more than $2$ to any other target points. Agents $\init{x}_i$ for $i>n$ can only reach to $\vec{p}_i$. To do so, these agents move in a gaming dimension and do not become truly qualified. 

Finally, we show how the solutions of these two problems coincide. Consider the problem of maximizing the true positives subject to including at most $k$ false positives. Including each $\vec{p}_j$ in the final set of target points, $\pfinal$, causes exactly one agent, $\init{x}_j$, to be a false positive. Therefore, having at most $k$ false positive is equivalent to including at most $k$ target points. Maximizing the true positives subject to at most $k$ target points is exactly equivalent to selecting at most $k$ sets that maximize the elements they cover. This completes the reduction.
\end{proof}

\begin{proof}[Proof of \Cref{thm:relaxed-no-destination-pts}]
We show the following problem is NP-hard.

\begin{problem}
\label{pr:relaxed-no-destination-pts}
Suppose we are given a set of $n$ agents where $\init{x}_1, \init{x}_2, \ldots, \init{x}_n$ denote their initial feature vectors.
Does there exist a set of target points $\pfinal \subseteq \mathbb{R}^d$ for which all the agents become truly qualified?
\end{problem}

We prove the NP-hardness by a reduction from the approximate version of the hitting set with equal-sized sets problem. As an instance of the hitting set problem we are given, $(\mathcal{F}, \mathcal{E})$ where $\mathcal{F} = \{S_1, \cdots, S_m\}$ is a collection of the subsets of $\mathcal{E}=\{e_1,e_2,\cdots,e_n\}$, and each set $S_i$ has a size of $0 < s < n$, and our goal is to find a minimum size set $S^*\subseteq \mathcal{E}$ that intersects every set in $\mathcal{F}$.  
In order to show NP-hardness, we construct an instance of \Cref{pr:relaxed-no-destination-pts} and prove: (1) If all the agents can become true positives by reaching to a set of target points $\pfinal\subseteq \mathbb{R}^d$ that the mechanism designer selects, then there exists a hitting set of size at most $2k$. (2) If there exists a hitting set of size $k$ then the mechanism designer can select a set of target points that encourages all the agents to become true positives. 
Since hitting set and set cover problems are equivalent and approximating set cover within a constant factor is NP-hard~\cite{feige1998threshold}, this implies that \Cref{pr:relaxed-no-destination-pts} is NP-hard.

First, we show how to construct an instance of \Cref{pr:relaxed-no-destination-pts} from an instance of the Hitting Set problem. To do so, we determine the number of dimensions, initial positions of the agents, the movement costs, and a linear threshold function for the truly qualified. Let $n$ be the number of elements of the Hitting Set instance, we construct an $n+1$-dimensional space where the first $n$ dimensions are improvement and the last dimension is gaming. Consider sets $S_1, S_2, \ldots, S_m$ in the Hitting Set instance. For every set $S_i$, we consider agent $i$ at initial point $\init{x}_i$. In $\init{x}_i$, the $j^{th}$ coordinates such that $e_j \in S_i$ is set to $1$. The rest of the first $n$ coordinates are set to $2k$ and the last coordinate is $0$. Also consider an extra agent $m+1$ at initial point $\init{x}_{m+1}$ where all the first $n$ coordinates are $0$ and the last coordinate is $2k(n-s)+s$. Note that for all the agents $\sum_{j=1}^{n+1} \init{x}_i[j] = 2k(n-s)+s$. Let the movement cost in all the dimensions $1 \leq j \leq n$ be  $\frac{1}{2k}$ and in dimension $n+1$ be $c$ such that $\frac{1}{2k(n-s)+s+1}<c<\frac{1}{2k(n-s)+s}$. Let $\f^*:\sum_{j=1}^{n+1} \vec{x}[j] \geq 2k(n-s)+s+2k$. Therefore, all the agents are initially unqualified and at $\ell_1$ distance of $2k$ from $\f^*$.

Now we prove the first direction, i.e., if all the agents can become true positives by reaching to a set of target points $\pfinal\subseteq \mathbb{R}^d$ that the mechanism designer selects, then there exists a hitting set of size at most $2k$. For all $1\leq i\leq m+1$, let $\vec{p}_i \in \pfinal$ denote the target point that $\init{x}_i$ moves to and becomes true positive.

It consists of the following arguments: (i) For all $1\leq i\leq m+1$, agent $i$ receives utility $0$ by reaching to $\vec{p}_i$. (ii) For all $1\leq i\leq m$, agent $m+1$ does not afford to reach to $\vec{p}_i$. 
(iii) If $\vec{p}_{m+1}[j] \leq 1$ for all $e_j \in S_i$, agent $i$ moves to $\vec{p}_{m+1}$ and becomes a false positive. Therefore if all agents improve, for each $1 \leq i \leq m$, there exists $e_j \in S_i$ such that $\vec{p}_{m+1}[j] > 1$. (iv) In order for agent $m+1$ to afford to reach to target point $\vec{p}_{m+1}$, the number of coordinates $1 \leq j \leq m$ with value at least $1$ must be at most $2k$. (v) These elements constitute a hitting set of size at most $2k$.

First, we prove argument (i). Each agent $1 \leq i \leq m+1$, is at $\ell_1$ distance of $2k$ to $\f^*$. To become qualified it needs to move $2k$ in the improvement dimensions.
Since moving for a distance of $2k$ along the improvement dimensions costs a value of $(2k)\times(\frac{1}{2k})=1$, agent $i$ makes a utility of $0$. 

Now, we move to argument (ii). Following up on the previous claim, to reach $\vec{p}_i$, agent $1 \leq i \leq m$ spends all of their movement budget in the improvement dimensions and cannot move a positive amount in the gaming dimension $n+1$. Therefore, $\vec{p}_i[n+1] = 0$ and $\sum_{j=1}^{n} \vec{p}_i[j] = 2k(n-s)+s+2k$. In order for agent $m+1$ to reach such a target point, it needs to move a total of $2k(n-s)+s+2k > 2k$ in the improvement dimensions, which costs more than $1$ and it cannot afford.

Next, we prove argument (iii). Since $\init{x}_{m+1}$ has an $\ell_1$ distance of $2k$ from $f^*$ and costs exactly a value of $1$ to reach there, it can only afford to move along the improvement dimensions. Therefore, $\vec{p}_{m+1}[n+1] \leq 2k(n-s)+s$.
Additionally, for $1 \leq j \leq n$, $\vec{p}_{m+1}[j] \leq 2k$; otherwise, agent $m+1$ cannot afford to reach to $\vec{p}_{m+1}$. Suppose $\vec{p}_{m+1}[j] \leq 1$ for all $e_j \in S_i$. Using this assumption, for agent $i$ to reach $\vec{p}_{m+1}$ it only needs to pay cost of movement in dimension $n+1$, moving $2k(n-s)+s$ units and paying $c$ per unit of movement. Since $(2k(n-s)+s)\times c < 1$, agent $i$ makes a strictly positive utility. Therefore agent $i$ prefers $\vec{p}_{m+1}$ over any other target point that makes it true positive which by argument (i) achieves utility $0$.

Argument (iv) is straight-forward. To achieve non-negative utility each agent can afford to move at most $2k$ units along the improvement dimensions. Therefore, for the target point  $\vec{p}_{m+1}$, the number of coordinates $1 \leq j \leq n$ with value at least $1$ must be at most $2k$.

Argument (v) is a direct implication of the two previous arguments. By argument (iii), for each $1 \leq i \leq m$ there is an element $e_j \in S_i$ such that $\vec{p}_{m+1}[j] > 1$. By argument (iv), the number of coordinates $j \leq n$ such that $\vec{p}_{m+1}[j] > 1$ is at most $2k$ since otherwise agent $m+1$ cannot afford to reach to $\vec{p}_{m+1}$. Therefore, elements $e_j$ such that $\vec{p}_{m+1}[j] > 1$ constitute a hitting set of size at most $2k$.

Now, we prove the reverse direction: if there exists a hitting set $S^*$ of size $k$, the mechanism designer can select a set of target points that encourages all the agents to become true positives. To do so, we construct a set of target points $\pfinal = \{\vec{p}_1, \ldots, \vec{p}_{m+1}\}$ that makes every agent to become true positive. For each agent $i$, $1 \leq i \leq m$, put a target point $\vec{p}_i$ whose first coordinate is $2k$ more than $\init{x}_i$. For agent $m+1$, put a target point $\vec{p}_{m+1}$ whose coordinates $j$ where $e_j \in S^*$ are set to $2$ and the remaining agree with $\init{x}_{m+1}$. Each target point $\init{x}_i$ is set such that $\sum_{j=1}^{n+1} \init{x}_i[j] = 2k(n-s)+s+2k$. In order to show that every agent is able to improve, we argue that: (i) For all $1\leq i \leq m$, agent $i$ can afford to move to $\vec{p}_i$. Additionally, if agent $i$ moves to any of the target points $\vec{p}_j$ where $1\leq j\leq m$, it becomes true positive.  (ii) For all $1\leq i\leq m$, agent $i$ cannot reach to $\vec{p}_{m+1}$. (iii) Agent $m+1$ moves to $\vec{p}_{m+1}$ and becomes true positive. 

First, we prove argument (i): Agent $i$ is at a distance of $2k$ from $\vec{p}_i$. It can afford to reach to $\vec{p}_i$ by paying a cost of $(2k)\times(\frac{1}{2k})=1$ and become true positive. In addition, if it moves to any of the other target points $\vec{p}_j$ where $1\leq j\leq m$, since it has only moved along the improvement dimensions, it would become true positive.

Next, we prove argument (ii): We know that for each $S_i$, there exists an element $e_j\in S_i$ such that $\vec{p}_{m+1}[j]=2$. As a result, the $\ell_1$ distance of $\init{x}_i$ and $\vec{p}_{m+1}$ is at least $(2k(n-s)+s+1)\times c > 1$. Therefore, for each $1\leq i\leq m$, $\init{x}_i$ cannot afford to reach to $\vec{p}_{m+1}$.

Finally, we prove argument (iii): First, we argue that agent $m+1$ cannot afford to reach to any of the target points $\vec{p}_i$ where $1\leq i\leq m$. For each target point $\vec{p}_i$ where $1\leq i\leq m$, $\vec{p}_i[n+1] = 0$ and $\sum_{j=1}^{n} \vec{p}_i[j] = 2k(n-s)+s+2k$. In order for agent $m+1$ to reach such a target point, it needs to move a total of $2k(n-s)+s+2k > 2k$ units in the improvement dimensions, which costs more than $1$ and it cannot afford. In addition, agent $m+1$ can afford to move to $\vec{p}_{m+1}$, and by reaching there it becomes true positive.

As a result of the above arguments, given a hitting set of size $k$, the mechanism designer can select a set of target points that encourages all the agents to become true positives.

Combining the above two directions, shows that the problem of selecting a set of target points for which all the agents become truly qualified is NP-hard.

\end{proof}

\section{Missing Proofs of \Cref{sec:two-dimension}}
\subsection{Proof of \Cref{lemma:2-dim-triangle-inequalities}}
\label{appendix:proof-lemma-2-dim-triangle-inequalities}
\begin{proof}
Initially, if $\vec{p}[2]< \init{x}_{i}[2]$, $\vec{p}$ is replaced with $(\vec{p}[1], \init{x}_i[2])$. By doing so, $cost(\init{x}_j, \vec{p})$ would not decrease. Hence, without loss of generality, we can assume $\vec{p}[2]\geq \init{x}_{i}[2]$.

First, we show that $cost(\init{x}_j,\vec{p})\leq cost(\init{x}_j,\vec{x}_{i,min})+cost(\vec{x}_{i,min},\vec{p})$, where the inequality holds when $\vec{x}_{i,min}[1]<\init{x}_j[1]\leq \vec{p}[1]$.
\begin{align*}
&cost(\init{x}_j,\vec{x}_{i,min})+cost(\vec{x}_{i,min},\vec{p})\\
&=max\Big\{\vec{x}_{i,min}[1]-\init{x}_{j}[1],0\Big\}+\Big(\vec{x}_{i,min}[2]-\init{x}_j[2]\Big)+\Big(\vec{p}[1]-\vec{x}_{i,min}[1]\Big)+\Big(\vec{p}[2]-\vec{x}_{i,min}[2]\Big)\\
&=max\Big\{\vec{x}_{i,min}[1]-\init{x}_{j}[1],0\Big\}+\Big(\vec{p}[1]-\vec{x}_{i,min}[1]\Big)+\Big(\vec{p}[2]-\init{x}_j[2]\Big)
\end{align*}
If $\init{x}_j[1]\leq \vec{x}_{i,min}[1]$, the last equation above gets equal to $\Big(\vec{p}[1]-\init{x}_j[1]\Big)+\Big(\vec{p}[2]-\init{x}_j[2]\Big) = cost(\init{x}_j,\vec{p})$. Otherwise, $\init{x}_j[1] > \vec{x}_{i,min}[1]$ and the last equation above gets equal to $\Big(\vec{p}[1]-\vec{x}_{i,min}[1]\Big)+\Big(\vec{p}[2]-\init{x}_j[2]\Big)> \Big(\vec{p}[1]-\init{x}_{j}[1]\Big)+\Big(\vec{p}[2]-\init{x}_j[2]\Big) = cost(\init{x}_j,\vec{p})$. In any case, $cost(\init{x}_j,\vec{p})\leq cost(\init{x}_j,\vec{x}_{i,min})+cost(\vec{x}_{i,min},\vec{p})$.

Next we argue that $cost(\vec{x}_{i,min},\vec{p})= cost(\vec{x}_{i,min},\vec{q})$. First, since $\vec{p}[1]\geq \vec{x}_{i,min}[1]$ and $\vec{p}[2]\geq \vec{x}_{i,min}[2]$, then $cost(\vec{x}_i,\vec{p}) = cost(\vec{x}_i,\vec{x}_{i,min})+ cost(\vec{x}_{i,min},\vec{p})$. Similarly, $cost(\vec{x}_i,\vec{q}) = cost(\vec{x}_i,\vec{x}_{i,min})+ cost(\vec{x}_{i,min},\vec{q})$. Since $cost(\vec{x}_i, \vec{p}) = cost(\vec{x}_i, \vec{q})$, it is the case that $cost(\vec{x}_{i,min},\vec{p})= cost(\vec{x}_{i,min},\vec{q})$.

Therefore,
\begin{align*}
& cost(\init{x}_j,\vec{p})\leq cost(\init{x}_j,\vec{x}_{i,min})+cost(\vec{x}_{i,min},\vec{p})\\
&\leq cost(\init{x}_j,\vec{x}_{i,min})+cost(\vec{x}_{i,min},\vec{q})\\
&=max\Big\{\vec{x}_{i,min}[1]-\init{x}_j[1], 0\Big\}+ max\Big\{\vec{x}_{i,min}[2]-\init{x}_j[2], 0\Big\}+\\
&max\Big\{\vec{q}[1]-\vec{x}_{i,min}[1],0\Big\}+max\Big\{\vec{q}[2]-\vec{x}_{i,min}[2],0\Big\}\\
&=max\Big\{\vec{x}_{i,min}[1]-\init{x}_j[1], 0\Big\}+ \Big(\vec{x}_{i,min}[2]-\init{x}_j[2]\Big)+\Big(\vec{q}[2]-\vec{x}_{i,min}[2]\Big)\\
&=max\Big\{\vec{x}_{i,min}[1]-\init{x}_j[1], 0\Big\}+ \Big(\vec{q}[2]-\init{x}_j[2]\Big)\\
&=max\Big\{\vec{q}[1]-\init{x}_j[1], 0\Big\}+ \Big(\vec{q}[2]-\init{x}_j[2]\Big)\\
&= cost(\init{x}_j, \vec{q})
\end{align*}
\end{proof}

\end{document}